\documentclass[12pt]{article}
\usepackage{amsthm}
\usepackage{amsfonts}
\usepackage{amsmath}
\usepackage{epsfig}
\usepackage{a4}
\newtheorem{theorem}{Theorem}
\newtheorem{conjecture}{Conjecture}
\newtheorem{lemma}[theorem]{Lemma}

\newtheorem{observation}[theorem]{Observation}
\newtheorem{proposition}[theorem]{Proposition}
\newcommand{\exc}{{\rm exc}}
\newcommand{\minexc}{{\rm minexc}}
\newcommand{\tsp}{{\rm tsp}}
\newcommand{\Prob}{{\mathbb P}}
\newcommand{\EE}{{\mathbb E}\;}
\newcommand{\RR}{{\mathbb R}}

\title{Graphic TSP in cubic graphs}

\author{Zden\v ek Dvo\v r\'ak\thanks{Charles University, Prague, Czech Republic.
E-mail: {\tt rakdver@iuuk.mff.cuni.cz}. Supported by the Center of Excellence -- Inst.~for Theor.~Comp.~Sci., Prague, project P202/12/G061 of Czech Science Foundation.}\and
Daniel Kr{\'a}l'\thanks{Mathematics Institute, DIMAP and Department of Computer Science, University of Warwick, Coventry CV4 7AL, UK. E-mail: {\tt d.kral@warwick.ac.uk}. This work has received funding from the European Research Council (ERC) under the European Union’s Horizon 2020 research and innovation programme (grant agreement No 648509). This publication reflects only its authors' view; the European Research Council Executive Agency is not responsible for any use that may be made of the information it contains.}\and
Bojan Mohar\thanks{Department of Mathematics, Simon Fraser University, Burnaby, B.C. V5A 1S6.
E-mail: {\tt mohar@sfu.ca}.  Supported in part by an NSERC Discovery Grant (Canada),
by the Canada Research Chairs program, and by a Research Grant of ARRS (Slovenia).}
}
\date{}

\begin{document}
\maketitle

\begin{abstract}
We present a polynomial-time $9/7$-approximation algorithm for the graphic TSP for cubic graphs,
which improves the previously best approximation factor of $1.3$ for $2$-connected cubic graphs and
drops the requirement of $2$-connectivity at the same time.
To design our algorithm,
we prove that every simple $2$-connected cubic $n$-vertex graph contains
a spanning closed walk of length at most $9n/7-1$, and that such a walk can be found in polynomial time.
\end{abstract}

\section{Introduction}

The Travelling Salesperson Problem (TSP) is one of the most central problems in combinatorial optimization.
The problem asks to find a shortest closed walk visiting each vertex at least once in an edge-weighted graph, or
alternatively to find a shortest Hamilton cycle in a complete graph where the edge weights satisfy the triangle inequality.
The Travelling Salesperson Problem is notoriously hard.
The approximation factor of $3/2$ established by Christofides~\cite{christo} has not been improved for 40 years
despite a significant effort of many researchers.
The particular case of the problem, the Hamilton Cycle Problem, was among the first problems to be shown to be NP-hard.
Moreover, Karpinski, Lampis and Schmied~\cite{KLS13} have recently shown that
the Travelling Salesperson Problem is NP-hard to approximate within the factor $123/122$,
improving the earlier inapproximability results of Lampis~\cite{cor13} and of Papadimitriou and Vempala~\cite{papavemp}.
In this paper, we are concerned with an important special case of the Travelling Salesperson Problem, the graphic TSP,
which asks to find a shortest closed walk visiting each vertex at least once in a graph where all edges have unit weight.
We will refer to such a walk as to a \emph{TSP walk}.

There have recently been a lot of research focused on approximation algorithms for the graphic TSP,
which was ignited by the breakthrough of the $3/2$-approximation barrier in the case of $3$-connected
cubic graphs by Gamarnik, Lewenstein and Sviridenko~\cite{cor08}.
This was followed by the improvement of the $3/2$-approximation factor for the general graphic TSP
by Oveis Gharan, Saberi and Singh~\cite{cor16}.
Next, M\"omke and Svensson~\cite{graprox} designed a $1.461$-approximation algorithm for the problem and
Mucha~\cite{cor15} showed that their algorithm is actually a $13/9$-approximation algorithm.
This line of research culminated with the $7/5$-approximation algorithm of Seb\"o and Vygen~\cite{cor21}.

We here focus on the case of graphic TSP for cubic graphs,
which was at the beginning of this line of improvements.
The $(3/2-5/389)$-approximation algorithm of Gamarnik et al.~\cite{cor08} for $3$-connected cubic graphs
was improved by Aggarwal, Garg and Gupta~\cite{cor01}, who designed a $4/3$-approximation algorithm.
Next, Boyd et al.~\cite{boyd} found a $4/3$-approximation algorithm for $2$-connected cubic graphs.
The barrier of the $4/3$-approximation factor was broken by Correa, Larr\'e and Soto~\cite{correa}
who designed a $(4/3-1/61236)$-approximation algorithm for this class of graphs.
The currently best algorithm for $2$-connected cubic graphs
is the $1.3$-approximation algorithm of Candr\'akov\'a and Lukot{\hskip -0.3ex}'ka~\cite{candluk},
based on their result on the existence of a TSP walk of length at most $1.3n-2$ in $2$-connected cubic $n$-vertex graphs.
We improve this result as follows.
Note that we obtain a better approximation factor and
Theorem~\ref{thm-alg} also applies to a larger class of graphs.

\begin{theorem}
\label{thm-main}
There exists a polynomial-time algorithm that for a given $2$-connected subcubic $n$-vertex graph
with $n_2$ vertices of degree two outputs a TSP walk of length at most
$$\frac{9}{7}n+\frac{2}{7}n_2-1\;.$$
\end{theorem}

\begin{theorem}
\label{thm-alg}
There exists a polynomial-time $9/7$-approximation algorithm for the graphic TSP for cubic graphs.
\end{theorem}

At this point, we should remark that we have not attempted to optimize the running time of our algorithm.
Also note that our approximation factor matches the approximation factor for cubic bipartite graphs
in the algorithm Karp and Ravi~\cite{karpra},
who designed a $9/7$-approximation algorithm for the graphic TSP for cubic bipartite graphs.
However, van Zuylen~\cite{zuylen} has recently found a $5/4$-approximation algorithm for this class of graphs.
Both the result of Karp and Ravi, and the result of van Zuylen are based on finding
a TSP walk of length of at most $9n/7$ and $5n/4$, respectively, in an $n$-vertex cubic bipartite graph.
On the negative side, Karpinski and Schmied~\cite{KS13} showed that
the graphic TSP is NP-hard to approximate within the factor of $535/534$ in the general case and
within the factor $1153/1152$ in the case of cubic graphs.

Our contribution in addition to improving the approximation factor for graphic TSP for cubic graphs
is also in bringing several new ideas to the table.
The proof of our main result, Theorem~\ref{thm-main}, differs from the usual line of proofs in this area.
In particular,
to establish the existence of a TSP walk of length at most $9n/7-1$ in a $2$-connected cubic $n$-vertex graph,
we allow subcubic graphs as inputs and perform reductions in this larger class of graphs.
While we cannot establish the approximation factor of $9/7$ for this larger class of graphs,
we are still able to show that
our techniques yields the existence of a TSP walk of length at most $9n/7-1$ for cubic $n$-vertex graphs.
In addition, unlike in the earlier results,
we do not construct a TSP walk in the final reduced graph by linking cycles in a spanning $2$-regular subgraph of the reduced graph
but we consider spanning subgraphs with vertices of degree zero and two, which gives us additional freedom.

We conclude with a brief discussion on possible improvements of the bound from Theorem~\ref{thm-main}.
In Section~\ref{sec-lb}, we give a construction of a $2$-connected cubic $n$-vertex graph
with no TSP walks of length smaller than $\frac{5}{4}n-2$ (Proposition~\ref{prop-qrepl}) and
a $2$-connected subcubic $n$-vertex graph with $n_2=\Theta(n)$ vertices of degree two
with no TSP walks of length smaller than $\frac{5}{4}n+\frac{1}{4}n_2-1$ (Proposition~\ref{prop-drepl});
the former construction was also found independently by Maz\'ak and Lukot{\hskip -0.3ex}'ka~\cite{mazluk}.
We believe that these two constructions provide the tight examples for an improvement of Theorem~\ref{thm-main} and
conjecture the following. We also refer to a more detailed discussion at the end of Section~\ref{sec-lb}.

\begin{conjecture}\label{conj-main}
Every $2$-connected subcubic $n$-vertex graph with $n_2$ vertices of degree
has a TSP walk of length at most
$$\frac{5}{4}n+\frac{1}{4}n_2-1\;.$$
\end{conjecture}

We would like to stress that it is important that Conjecture~\ref{conj-main} deals with simple graphs,
i.e., graphs without parallel edges. Indeed, consider the cubic graph $G$ obtained as follows:
start with the graph that has two vertices of degree three that are joined by three paths,
each having $2\ell$ internal vertices of degree two, and replace every second edge of these paths with a pair of parallel edges
to get a cubic graph. The graph $G$ has $n=6\ell+2$ vertices but no TSP walk of length shorter than $8\ell+2$.

\section{Preliminaries}\label{sec-prelim}

In this section, we fix the notation used in the paper and
make several simple observations on the concepts that we use.

All graphs considered in this paper are \emph{simple}, i.e., they do not contain parallel edges.
When we allow parallel edges, we will always emphasize this by referring to a considered graph as to a \emph{multigraph}.
We will occasionally want to stress that a graph obtained during the proof has no parallel edges and
we will do so by saying that it is simple even if saying so is superfluous.
The \emph{underlying graph} of a multigraph $H$ is the graph obtained
from $H$ by suppressing parallel edges,
i.e., replacing each set of parallel edges by a single edge.

If $G$ is a graph, its vertex set is denoted by $V(G)$ and its edge set by $E(G)$.
Further, the number of vertices of $G$ is denoted by $n(G)$ and the number of its vertices of degree two by $n_2(G)$.
If $w$ a vertex of $G$,
then $G-w$ is a graph obtained by deleting the vertex $w$ and all the edges incident with $w$.
Similarly, if $W$ is a set of vertices of $G$, then $G-W$ is the graph obtained by deleting
all vertices of $W$ and edges incident with them.
Finally, if $F$ is a set of its edges, then $G\setminus F$ is the graph obtained from $G$ by removing the edges of $F$
but none of the vertices.

A graph with all vertices of degree at most three is called \emph{subcubic}.
We say that a graph $G$ is \emph{$k$-connected}
if it has at least $k+1$ vertices and $G-W$ is connected for any $W\subseteq V(G)$ containing at most $k-1$ vertices.
If $G$ is connected but not $2$-connected, then a vertex $v$ such that $G-v$ is not connected is called a \emph{cut-vertex}.
Maximal $2$-connected subgraphs of $G$ are called \emph{blocks}.
Note that a vertex of a graph is contained in two or more blocks if and only if it is a cut-vertex.
A subset $F$ of the edges of a graph $G$ is an \emph{edge-cut}
if the graph $G\setminus F$ have more components than $G$ and $F$ is minimal with this property.
Such a subset $F$ containing exactly $k$ edges will also be referred to as \emph{$k$-edge-cut}.
An edge forming a $1$-edge-cut is called a \emph{cut-edge}.
A graph $G$ is \emph{$k$-edge-connected} if it has no $\ell$-edge-cut for $\ell\le k$.
Note that a subcubic graph $G$ with at least two vertices is $2$-connected if and only if $2$-edge-connected.

A \emph{$\theta$-graph} is a simple graph obtained from the pair of vertices joined by three parallel edges
by subdividing some of the edges several times.
In other words,
a $\theta$-graph is a graph that contains two vertices of degree three joined by three paths formed by vertices of degree two
such that at most one of these paths is trivial, i.e., it is a single edge.
In our consideration, we will need to consider a special type of cycles of length six in subcubic graphs,
which resembles $\theta$-graphs.
A cycle $K=v_1\ldots v_6$ of length six in a subcubic graph $G$ is a \emph{$\theta$-cycle},
if all vertices $v_1,\ldots,v_6$ have degree three, their neighbors $x_1$, \ldots, $x_6$ outside of $K$ are pairwise distinct, and
if $G-V(K)$ has three connected components,
one containing $x_1$ and $x_2$, one containing $x_4$ and $x_5$, and one containing $x_3$ and $x_6$.
See Figure \ref{fig-thetacycle} for an example.
The vertices $v_3$ and $v_6$ of the cycle $K$ will be referred to as the \emph{poles} of the $\theta$-cycle $K$.

\begin{figure}
\begin{center}
\includegraphics[width=60mm]{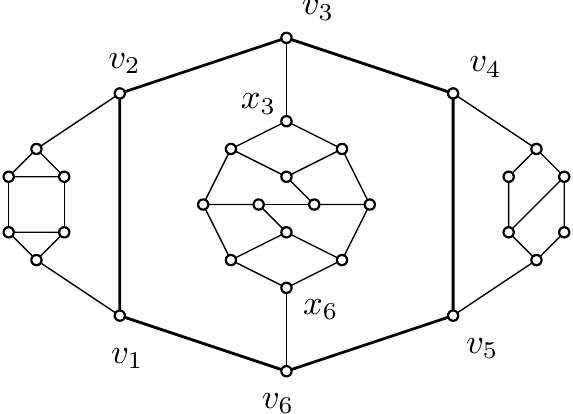}
\end{center}
\caption{A $\theta$-cycle with poles $v_3$ and $v_6$.}\label{fig-thetacycle}
\end{figure}

We say that a multigraph is \emph{Eulerian} if all its vertices have even degree;
note that we do not require the multigraph to be connected,
i.e., a multigraph has an Eulerian tour if and only if it is Eulerian and connected.
A subgraph is \emph{spanning} if it contains all vertices of the original graphs,
possibly some of them as isolated vertices, i.e., vertices of degree zero.
It is easy to relate the length of the shortest TSP walk in a graph $G$
to the size of Eulerian multigraphs using edges of $G$ as follows.
To simplify our presentation, let $\tsp(G)$ denote the length of the shortest TSP walk in a graph $G$.

\begin{observation}\label{obs-ep}
For every graph $G$,
$\tsp(G)$ is equal to the minimum number of edges of a connected Eulerian multigraph $H$ 
such that the underlying graph of $H$ is a spanning subgraph of $G$.
\end{observation}

\begin{proof}
Let $W$ be a TSP walk of length $\tsp(G)$, and
let $H$ be the multigraph on the same vertex set as $G$ such that
each edge $e$ of $G$ is included to $H$ with multiplicity equal to the number of times that it is used by $W$.
In particular, edges not traversed by $W$ are not included to $H$ at all.
Clearly, the multigraph $H$ is connected and Eulerian,
the number of its edges is equal to the length of $W$ and
its underlying graph is a spanning subgraph of $G$.

We next establish the other inequality claimed in the statement.
Let $H$ be a connected Eulerian multigraph whose underlying graph is a spanning subgraph of $G$, and
$H$ has the smallest possible number of edges.
A closed Eulerian tour in $H$ yields a TSP walk in $G$ (just follow the tour in $G$) and
the length of this TSP walk is equal to the number of edges of $H$.
Hence, $\tsp(G)$ is at most the number of edges of $H$.
\end{proof}

We now explore the link between Eulerian spanning subgraphs and
the minimum length of a TSP walk further.
For a graph $G$, let $c(F)$ denote the number of non-trivial components of $F$,
i.e., components formed by two or more vertices, and
let $i(F)$ be the number of isolated vertices of $F$.
We define the \emph{excess} of a graph $F$ as
$$\exc(F)=2c(F)+i(F).$$
If $G$ is a subcubic graph, we define 
$$\minexc(G)=\min\;\{\exc(F):\text{$F$ spanning Eulerian subgraph of $G$}\}.$$
Note that any subcubic Eulerian graph $F$ is a union of $c(F)$ cycles and $i(F)$ isolated vertices,
i.e., the spanning subgraph $F$ of a subcubic graph $G$ with $\exc(F)=\minexc(G)$ must also have this structure.
The values of $\minexc(G)$ for simple-structured graphs are given in the next observation (note that
the condition $k_2\not=0$ implies that the $\theta$-graph is simple).

\begin{observation}
\label{obs-nontarg}
The following holds.
\begin{enumerate}
\item If $G$ is a cycle, then $\minexc(G)=2<\frac{n(G)+n_2(G)}{4}+1$.
\item If $G=K_4$, then $\minexc(G)=2=\frac{n(G)+n_2(G))}{4}+1$.
\item If $G$ is $\theta$-graph with $k_1$, $k_2$ and $k_3$ vertices of degree two on the paths joining its two
      vertices of degree three and $k_1\le k_2\le k_3$ and $k_2\not=0$, $\minexc(G)=2+k_1\le\frac{n(G)+n_2(G)}{4}+1$.
\end{enumerate}
\end{observation}

We next relate the quantity $\minexc(G)$ to the length of the shortest TSP walk in $G$.

\begin{observation}\label{obs-exc}
Let $G$ be a connected subcubic $n$-vertex graph, and
let $F$ be a spanning Eulerian subgraph $F$ of $G$.
There exists a polynomial-time algorithm that finds a TSP walk of length $n-2+\exc(F)$.
In addition, the minimum length of a TSP walk in $G$ is equal to
$$\tsp(G)=n-2+\minexc(G)\;.$$
\end{observation}

\begin{proof}
Let $F$ be a spanning Eulerian subgraph of $G$.
We aim to construct a TSP walk of length $n-2+\exc(F)$.
The subgraph $F$ has $c(F)+i(F)$ components.
Since $F$ is subcubic, each of the $c(F)$ non-trivial components of $F$ is a cycle,
which implies that $F$ has $n-i(F)$ edges.
Since $G$ is connected, there exists a subset $S$ of the edges of $G$ such that $|S|=c(F)+i(F)-1$ and
$F$ together with the edges of $S$ is connected.
Clearly, such a subset $S$ can be found in linear time.
Let $H$ be the multigraph obtained from $F$ by adding each edge of $S$ with multiplicity two.
Since $H$ is a connected Eulerian multigraph whose underlying graph is a spanning subgraph of $G$,
the proof of Observation~\ref{obs-ep} yields that it corresponds to an Eulerian tour of length
$$|E(H)|=|E(F)|+2|S|=n-i(F)+2(c(F)+i(F)-1)=n-2+\exc(F),$$
which can be found in linear time.
In particular, it holds that $\tsp(G)\le n-2+\exc(F)$.
Since the choice of $F$ was arbitrary,
we conclude that $\tsp(G)\le n-2+\minexc(G)$.

To finish the proof, we need to show that $n-2+\minexc(G)\le\tsp(G)$.
By Observation~\ref{obs-ep}, there exists a connected Eulerian multigraph $H$ with $|E(H)|=\tsp(G)$
such that its underlying graph is a spanning subgraph of $G$.
By the minimality of $|E(H)|$, every edge of $H$ has multiplicity at most two (otherwise,
we can decrease its multiplicity by $2$ while keeping the multigraph Eulerian and connected).
Similarly, removing any pair of parallel edges of $H$ disconnects $H$ (as the resulting multigraph would still be Eulerian),
i.e., the edge in the underlying graph of $H$ corresponding to a pair of parallel edges is a cut-edge.
Let $F$ be the graph obtained from $H$ by removing all the pairs of parallel edges.
The number of components of $F$ is equal to
$$c(F)+i(F)=\frac{|E(H)|-|E(F)|}{2}+1.$$
Since $F$ is subcubic, it is a union of $c(F)$ cycles and $i(F)$ isolated vertices,
which implies that $|E(F)|=n-i(F)$.
Consequently, we get that
$$c(F)+i(F)=\frac{|E(H)|-(n-i(F))}{2}+1,$$
which yields the desired inequality
$$n-2+\minexc(G)\le n-2+\exc(F)=n-2+2c(F)+i(F)=|E(H)|=\tsp(G).$$
\end{proof}

\section{Reductions}\label{sec-redu}

In this section, we present a way of reducing a 2-connected subcubic graph to a smaller one such that
a spanning Eulerian subgraph of the smaller graph yields a spanning Eulerian subgraph of the original graph
with few edges.
We now define this process more formally.
For subcubic graphs $G$ and $G'$, let $$\delta(G,G')=(n(G)+n_2(G))-(n(G')+n_2(G'))\;.$$
We say that a 2-connected subcubic graph $G'$ is a reduction of a 2-connected subcubic graph $G$
if $n(G')<n(G)$, $\delta(G,G')\ge 0$, and there exists a linear-time algorithm that
turns any spanning Eulerian subgraph $F'$ of $G'$ into a spanning Eulerian subgraph $F$ of $G$ satisfying
\begin{equation}\label{eq-redu}
\exc(F)\le \exc(F')+\tfrac{\delta(G,G')}{4}.
\end{equation}
For the proof of our main result,
it would be enough to prove the lemmas in this section with $\frac{1}{4}$ replaced by $\frac{2}{7}$ in (\ref{eq-redu}).
However, this would not simplify most of our arguments and
we believe that the stronger form of (\ref{eq-redu}) can be useful in an eventual proof of Conjecture~\ref{conj-main}.

The reductions that we consider involve altering a subgraph $K$ of a graph $G$ such that
$K$ has some additional specific properties. This subgraph sometimes needs to be provided as
a part of an input of an algorithm that constructs $G'$.
We say that a reduction is a \emph{linear-time reduction with respect to a subgraph $K$}
if there exists a linear-time algorithm that transforms $G$ to $G'$ given $G$ and a subgraph $K$ with the specific properties.
We will say that a reduction is a \emph{linear-time reduction}
if there exists a linear-time algorithm that both finds a suitable subgraph $K$ and performs the reduction.
If a graph $G$ admits such a reduction, we will say that $G$ has a linear-time reduction or
that $G$ has a linear-time reduction with respect to a subgraph $K$.

The reductions that we present are intended to be applied to an input subcubic $2$-connected graph
until the resulting graph is simple or it becomes having a special structure.
A subcubic $2$-connected graph is \emph{basic} if it is a cycle, a $\theta$-graph, or $K_4$.
A subcubic $2$-connected graph that is not basic will be referred to as \emph{non-basic}.
We say that a $2$-connected subcubic graph $G$ is a \emph{proper} graph
if $G$ is non-basic, has no cycle with at most four vertices of degree three, and
has no cycle of length five or six with five vertices of degree three.
In Subsection~\ref{sub-proper},
we will show that every non-basic $2$-connected subcubic graph that is not proper has a linear-time reduction.
In addition to proper $2$-connected subcubic graph,
we will also consider clean $2$-connected subcubic graphs.
This definition is more complex and we postpone it to Subsection~\ref{sub-clean}.

\subsection{Cycles with few vertices of degree three}
\label{sub-proper}

In this subsection, we show that a non-basic $2$-connected subcubic graph that is not proper
has a linear-time reduction, i.e.,
every graph containing a cycle with at most four vertices of degree three or
a cycle of length five or six containing five vertices of degree three
has a linear-time reduction.
We present the reductions in Lemmas~\ref{lemma-c2e}--\ref{lemma-no5} assuming that such a cycle is given.
We remark that such a cycle can be found in linear time (if it exists) using the following argument:
a subcubic $n$-vertex graph $G$ has at most $3\cdot 2^{k-1}n$ cycles containing at most $k$ vertices of degree three.
Indeed, suppressing all vertices of degree two in $G$ results in a cubic multigraph,
its cycles of length at most $k$ one-to-one correspond to cycles with at most $k$ vertices of degree three in $G$, and
it is possible to list all cycles of length at most $k$ in a cubic multigraph in linear time.
The fact that we can list all such cycles in linear time is important for Lemmas~\ref{lemma-no4} and~\ref{lemma-no5}
where we need to choose a cycle with at most $k$ vertices of degree three with some additional properties.

\begin{lemma}\label{lemma-c2e}
Every non-basic $2$-connected subcubic graph $G$ that contains a cycle $K$ with at most two vertices of degree three
has a linear-time reduction.
\end{lemma}

\begin{proof}
Since $G$ is neither a cycle nor a $\theta$-graph, it follows that $V(G)\neq V(K)$.
Since $G$ is 2-connected, $K$ contains exactly two vertices of degree three, say $v_1$ and $v_2$.
Let $x_1$ and $x_2$ be their neighbors outside of $K$, and
let $k_1$ and $k_2$ be the the number of the internal vertices of the two paths between $v_1$ and $v_2$ in $K$.
We can assume that $k_1\le k_2$ by symmetry.
If $x_1=x_2$, then either $G$ is a $\theta$-graph or $x_1$ is incident with a cut-edge;
since neither of these is possible, it holds that $x_1\neq x_2$.

Suppose that $k_1=0$ and $k_2=1$, i.e., $K$ is a triangle.
Let $z$ be the vertex of $K$ distinct from $v_1$ and $v_2$, and let $G'=G-z$.
Note that $G'$ is a 2-connected subcubic graph.
We claim that $G'$ is a reduction of $G$.
Since $n(G')=n(G)-1$ and $n_2(G')=n_2(G)+1$, it follows $\delta(G,G')=0$.
Consider a spanning Eulerian subgraph $F'$ of $G'$.
If $F'$ contains the edge $v_1v_2$,
then let $F$ be the spanning Eulerian subgraph of $G$ obtained from $F'$ by removing the edge $v_1v_2$ and adding the path $v_1zv_2$.
If $F'$ does not contain the edge $v_1v_2$, i.e., $v_1$ and $v_2$ are isolated vertices of $F$,
then let $F$ be the spanning Eulerian subgraph of $G$ obtained from $F'$ by adding the cycle $K$.
It holds that $\exc(F)=\exc(F')$ in both cases.

It remains to consider the case $k_1+k_2\ge 2$.
Let $G'$ be obtained from $G-V(K)$ by adding a path $x_1wx_2$ where $w$ is a new vertex;
note that $w$ has degree two in $G'$ and $\delta(G,G')=2(k_1+k_2)$.
Since $x_1\neq x_2$, $G'$ is simple.
We show that $G'$ is a reduction of $G$.
Let $F'$ be a spanning Eulerian subgraph of $G'$;
we will construct a spanning Eulerian subgraph $F$ of $G$.
If $F'$ contains the path $x_1wx_2$,
then let $F$ be obtained from $F'-w$ by adding the vertices of $K$ and the edges $x_1v_1$, $x_2v_2$, and
the path in $K$ between $v_1$ and $v_2$ with $k_2$ internal vertices.
Note that the $k_1$ vertices of the other path between $v_1$ and $v_2$ in $K$ are isolated in $F$.
Observe that
$$\exc(F)=\exc(F')+k_1\le \exc(F')+\frac{k_1+k_2}{2},$$
since $k_1\le k_2$.
If $w$ is an isolated vertex of $F'$,
then let $F$ be obtained from $F'-w$ by adding the cycle $K$.
In this case, we get that
$$\exc(F)=\exc(F')+1\le \exc(F')+\frac{k_1+k_2}{2}.$$
Since it holds that $\exc(F)\le\exc(F)+\frac{1}{4}\delta(G,G')$ in both cases,
the proof of the lemma is finished.
\end{proof}

In the next lemma, we consider cycles containing three vertices of degree three.

\begin{lemma}\label{lemma-no3}
Every non-basic $2$-connected subcubic graph $G$ that contains a cycle $K$ with three vertices of degree three
has a linear-time reduction.
\end{lemma}

\begin{proof}
Let $v_1$, $v_2$ and $v_3$ be the three vertices of degree three of $K$.
Since $G$ is 2-connected, each of the vertices $v_1$, $v_2$ and $v_3$ has a neighbor outside the cycle $K$;
let $x_i$ be such a neighbor of the vertex $v_i$, $i\in\{1,2,3\}$.
Further, let $P_i$ denote the path between $v_{i+1}$ and $v_{i+2}$ in $K$ that does not contain $v_i$ for $i\in\{1,2,3\}$ (indices
are taken modulo three), and let $k_i$ be the number of its internal vertices.
By symmetry, we can assume that $k_1\le k_2\le k_3$.
Since $G$ is not basic, in particular, $G\neq K_4$,
we can assume that $x_2\neq x_3$ if $k_1=k_2=k_3=0$.

Let $G'$ be obtained from $G-V(K)$ by adding a vertex $z$ and
paths $Q_1$, $Q_2$ and $Q_3$ joining $z$ with $x_1$, $x_2$, and $x_3$, respectively,
such that $Q_1$ has $k_1+1$ internal vertices, $Q_2$ has $k_2$ internal vertices, and $Q_3$ has $k_3$ internal vertices.
Note that the graph $G'$ is simple since if $k_2=k_3=0$, then $x_2\neq x_3$.
Also note that $\delta(G,G')=0$.

We now show that $G'$ is a reduction of $G$.
Let $F'$ be a spanning Eulerian subgraph of $G'$.
If the vertex $z$ is isolated in $F'$, then let $F$ be obtained from $F'$ by removing $z$ and
the internal vertices of $Q_1$, $Q_2$, and $Q_3$ (all of these are isolated vertices in $F'$) and adding the cycle $K$.
Observe that $c(F)=c(F')+1$ and $i(F)=i(F')-2-k_1-k_2-k_3$ in this case.
If $F'$ contains paths $Q_i$ and $Q_j$, $i\neq j$, $i,j\in\{1,2,3\}$,
then let $F$ be obtained from $F'$ by removing the vertex $z$ and the internal vertices of $Q_1$, $Q_2$, and $Q_3$,
adding the vertices of $K$, edges $x_iv_i$ and $x_jv_j$, and the edges of the paths $P_i$ and $P_j$.
We have $c(F)=c(F')$ and $i(F)\le i(F')$ in this case.
In both cases, it holds $\exc(F)\le\exc(F')$, which finishes the proof of the lemma.
\end{proof}

In the final two lemmas of this subsection,
we will present several possible reductions of a configuration $K$ and
choose the one that is $2$-connected.
Since it is possible to test $2$-connectivity of a graph in linear time,
the reductions presented in Lemmas~\ref{lemma-no4} and~\ref{lemma-no5} are linear-time.

\begin{lemma}\label{lemma-no4}
Every non-basic $2$-connected subcubic graph $G$ that contains a cycle $K$ with four vertices of degree three
has a linear-time reduction.
\end{lemma}

\begin{proof}
Choose a shortest cycle $K$ of $G$ that contains four vertices of degree four, and
let $v_1,\ldots,v_4$ be these vertices listed in the cyclic order around $K$.
Since $K$ is the shortest possible and
all cycles in $G$ contain at least four vertices of degree three by Lemmas~\ref{lemma-c2e} and~\ref{lemma-no3},
every vertex $v_i$ has a neighbor $x_i$ outside the cycle $K$, $i\in\{1,\ldots,4\}$.
In addition, it holds that $x_i\neq x_{i+1}$ (indices are taken modulo four).
Let $P_i$ denote the path between $v_i$ and $v_{i+1}$ in $K$ (again, indices are taken modulo four), and
let $k_i$ be the number of internal vertices of $P_i$.
Finally, let $k=k_1+\cdots+k_4$.

We present two possible reductions parameterized by $j\in\{1,2\}$.
Let $G_j$ be the graph obtained from $G$ by removing the edges and internal vertices of the paths $P_j$ and $P_{j+2}$.
Suppose that neither $G_1$ nor $G_2$ is 2-connected.
In particular, the vertices of $G_1$ can be partitioned into non-empty sets $A$ and $B$ such that
there is at most one edge between $A$ and $B$ of $G_1$.
If $x_1\in A$ and $x_4\in B$, then this edge is contained in $P_4+x_1v_1+x_4v_4$;
by symmetry, we can assume that $x_2,x_3\in B$,
which yields that the edge $v_1x_1$ is a cut-edge in $G$, which is impossible.
Hence, it must hold that $x_1,x_4\in A$ and $x_2,x_3\in B$.
Since $G$ is 2-connected, there exists a path between $x_1$ and $x_4$ using only the vertices of $A\setminus V(K)$ and
a path between $x_2$ and $x_3$ using only the vertices of $B\setminus V(K)$.
The symmetric argument applied to $G_2$ yields the existence of such paths
between $x_1$ and $x_2$, and between $x_3$ and $x_4$, which is impossible
since there is at most one edge between $A$ and $B$.
It follows that at least one of the graphs $G_1$ and $G_2$ is 2-connected.
By symmetry, we assume that $G_1$ is $2$-connected in the rest of the proof.

We first consider the case that $k_1=k_3=0$.
Let $G'$ be the graph obtained from $G-V(K)$ by adding paths $x_1z_1x_4$ and $x_2z_2x_3$,
where $z_1$ and $z_2$ are new vertices, each having degree two in $G'$.
Note that $G'$ is is isomorphic to a graph obtained from $G_1$ by suppressing some vertices of degree two;
in particular, $G'$ is 2-connected.
Also note that $\delta(G,G')=2k$.
We next show that $G'$ is a reduction of $G$.
Consider a spanning Eulerian subgraph $F'$ of $G'$.
We distinguish several cases based on whether the vertices $z_1$ and $z_2$ are isolated in $F'$.
\begin{itemize}
\item If both vertices $z_1$ and $z_2$ are isolated in $F'$,
      then let $F$ be obtained from $F'-\{z_1,z_2\}$ by adding the cycle $K$.
      Note that $\exc(F)=\exc(F')$ in this case.
\item Assume that $z_1$ is not isolated, i.e., the edges $x_1z_1$ and $x_4z_1$ are contained in $F'$,
      but $z_2$ is isolated in $F'$.
      We consider two spanning Eulerian subgraphs $F_1$ and $F_2$ of $G$.
      The subgraph $F_1$ is obtained from $F'-\{z_1,z_2\}$ by adding the vertices of $K$,
      the edges $x_1v_1$ and $x_4v_4$, and edges of the path $P_4$.
      The subgraph $F_2$ is obtained from $F'-\{z_1,z_2\}$ by adding the vertices of $K$,
      the edges $x_1v_1$ and $x_4v_4$, and the edges of the paths $P_1$, $P_2$, and $P_3$.
      Note that $\exc(F_1)=\exc(F')+k_1+k_2+k_3+1=\exc(F')+k_2+1$ and $\exc(F_2)=\exc(F')+k_4-1$.
      Let $F$ be one of the subgraphs $F_1$ and $F_2$ with the smaller excess.
      Since $\exc(F_1)+\exc(F_2)=2\exc(F')+k$, we get that $\exc(F)\le\exc(F')+k/2$.
\item The case that $z_1$ is isolated in $F'$ but $z_2$ is not is symmetric to the case that we have just analyzed.
\item If neither $z_1$ nor $z_2$ is isolated in $F'$,
      then let $F$ be obtained from $F'-\{z_1,z_2\}$ by adding the vertices of $K$,
      the edges $x_iv_i$ for $i\in\{1,\ldots, 4\}$, and
      the edges of the paths $P_2$ and $P_4$.
      Since $k_1=k_3=0$, we get that $\exc(F)=\exc(F')$.
\end{itemize}
In all the cases we have found a spanning Eulerian subgraph $F$ of $G$ with $\exc(F)\le\exc(F')+k/2=\exc(F')+\delta(G,G')/4$.

We can assume that $k_1+k_3\ge 1$ in the rest of the proof.
Note that this implies that neither $x_1x_4$ nor $x_2x_3$ is an edge of $G$ (otherwise,
$G$ would contain a cycle with at most four vertices of degree three that is shorter than $K$).

We now distinguish two cases: $k\ge 2$ and $k=1$.
We first consider the case that $k\ge 2$.
Let $G'$ be the graph obtained from $G-V(K)$ by adding edges $x_1x_4$ and $x_2x_3$.
Since $G'$ can be obtained from $G_1$ by suppressing vertices of degree two,
it follows that $G'$ is 2-connected.
Also note that $G'$ is simple since neither $x_1x_4$ nor $x_2x_3$ is an edge of $G$, and that $\delta(G,G')=2k+4$.
We next verify that $G'$ is a reduction of $G$.
To do so, consider a spanning Eulerian subgraph $F'$ of $G'$ and
distinguish four cases based on the inclusion of the edges $x_1x_4$ and $x_2x_3$ in $F'$
to construct a spanning Eulerian subgraph $F$ of $G$.
\begin{itemize}
\item If neither the edge $x_1x_4$ nor the edge $x_2x_3$ is in $F$,
      then let $F$ be obtained from $F'$ by adding the cycle $K$.
      Note that $\exc(F)=\exc(F')+2$.
\item If the edge $x_1x_4$ is in $F'$ but the edge $x_2x_3$ is not,
      then we consider two spanning Eulerian subgraphs $F_1$ and $F_2$ of $G$, and
      choose $F$ to be the one with the smaller excess.
      The subgraph $F_1$ is obtained from $F'$ by removing the edge $x_1x_4$ and
      by adding the vertices of $K$ and the edges $x_1v_1$ and $x_4v_4$, and the edges of the path $P_4$.
      The subgraph $F_2$ is obtained from $F'$ by removing the edge $x_1x_4$ and
      by adding the vertices of $K$ and the edges $x_1v_1$ and $x_4v_4$, and the edges of the paths $P_1$, $P_2$, and $P_3$.
      Note that $\exc(F_1)=\exc(F')+k_1+k_2+k_3+2$ and $\exc(F_2)=\exc(F')+k_4$.
      Hence, if $F$ is the one of the subgraphs $F_1$ and $F_2$ with the smaller excess,
      then $\exc(F)\le\exc(F')+\frac{k}{2}+1$.
\item The case that the edge $x_1x_4$ is not contained in $F'$ but the edge $x_2x_3$ is
      is symmetric to the case that we have just analyzed.
\item The final case is that both the edges $x_1x_4$ and $x_2x_3$ are in $F'$.
      We again construct two spanning Eulerian subgraphs $F_1$ and $F_2$ of $G$, and
      choose $F$ to be the one with the smaller excess.
      We start with removing the edges $x_1x_4$ and $x_2x_3$ from $F'$ and
      adding the vertices of $K$ together with the edges $x_iv_i$ for $i\in\{1,\ldots, 4\}$.
      To create the subgraph $F_1$, we also add the edges of the paths $P_2$ and $P_4$, and
      to create the subgraph $F_2$, we add the edges of the paths $P_1$ and $P_3$.
      Note that the latter can result in either creating or merging two cycles of $F'$,
      in particular, $c(F_2)\le c(F')+1$.
      Hence, we get that $\exc(F_1)=\exc(F')+k_1+k_3$ and $\exc(F_2)\le\exc(F')+k_2+k_4+2$.
      Since $F$ is the one of the subgraphs $F_1$ and $F_2$ with the smaller excess,
      we get that $\exc(F)\le\exc(F')+\frac{k}{2}+1$.
\end{itemize}
Since $k\ge 2$,
the excess $\exc(F)$ of the spanning Eulerian subgraph $F$ of $G$ is at most $\exc(F')+\frac{k}{2}+1=\exc(F')+\delta(G,G')/4$
in all the four cases.

The final case to consider is that $k=1$.
Since $k_1+k_3\ge 1$, we can assume by symmetry that $k_1=1$ and $k_2,k_3,k_4=0$.
In this case, we consider the graph $G'$ obtained from $G-V(K)$ by adding the edge $x_1x_4$ and a path $x_2zx_3$,
where $z$ is a new vertex of degree two.
Again, $G'$ is isomorphic to a graph obtained from $G_1$ by suppressing some vertices of degree two,
in particular, $G'$ is 2-connected. Also note that $\delta(G,G')=4$.
To show that $G'$ is a reduction of $G$,
one considers a spanning Eulerian subgraph $F'$ of $G'$ and distinguish four cases based on whether
the edge $x_1x_4$ and the path $x_2zx_3$ are contained in $F'$.
If neither of them is, we construct a spanning Eulerian subgraph $F$ of $G$ by removing the vertex $z$ and
including the cycle $K$; note that $\exc(F)=\exc(F')+1$ in this case.
If one of them but the other is not,
we construct a spanning Eulerian subgraph $F$ by removing the edge $x_1x_4$ and the edges of the path $x_2zx_3$,
adding the vertices of $K$ together with the edges $x_iv_i$ for those $i\in\{1,\ldots, 4\}$ such that the degree of $x_i$ is odd and
the edges of three of the paths $P_1$, $P_2$, $P_3$ and $P_4$ in a way that $F$ is an Eulerian subgraph of $G$.
Note that $\exc(F)\le\exc(F')$ (the inequality is strict if $z$ is an isolated vertex in $F'$)
Finally, if both the edge $x_1x_4$ and the path $x_2zx_3$ are contained in $F'$,
we construct $F$ by removing the edge $x_1x_4$ and the edges of the path $x_2zx_3$, and
by adding the vertices of $K$ together with the edges $x_iv_i$ for $i\in\{1,\ldots, 4\}$ and
the edges of the paths $P_2$ and $P_4$.
Note that $\exc(F)=\exc(F')+1$ in this case since the only inner vertex of $P_1$ is isolated in $F$.
Hence, we have constructed a spanning Eulerian subgraph $F$ of $G$
with $\exc(F)\le\exc(F')+1=\exc(F')+\delta(G,G')/4$ in each of the cases.
\end{proof}

In the final lemma of this subsection, we deal with cycles of length five or six that
contain five vertices of degree three.

\begin{lemma}\label{lemma-no5}
Every non-basic $2$-connected subcubic graph $G$ that contains a cycle $K$ of length at most $6$ with five vertices of degree three
has a linear-time reduction.
\end{lemma}

\begin{proof}
By Lemmas~\ref{lemma-c2e}--\ref{lemma-no4},
we can assume that every cycle of $G$ contains at least five vertices of degree three.
Let $K$ be a cycle of length five or six that contains five vertices of degree three.
If $G$ contains such cycles of length five or six, choose $K$ to be a cycle of length five.
By symmetry, we can assume that the vertices $v_1$, \ldots, $v_5$ of degree three of $K$ form a path $v_1v_2v_3v_4v_5$;
if $K$ has length five, then $v_5v_1$ is an edge, and
if $K$ has length six, then there is a vertex $z$ of degree two such that $v_5zv_1$ is a path in $G$.
Let $x_i$ be the neighbor of $v_i$ outside the cycle $K$ for $i\in\{1,\ldots,5\}$.
The vertices $x_1,\ldots,x_5$ are pairwise distinct (otherwise, $G$ would contain a cycle with at most four vertices of degree three).
Since $G$ has no cycle with at most four vertices of degree three,
$G$ does not contain the edge $x_ix_{i+1}$ for any $i\in\{1,\ldots,5\}$ (indices are taken modulo five).

Let $G_1$ be the graph obtained from $G-V(K)$ by adding the edge $x_5x_1$ and
a new vertex $w$ that is adjacent to the vertices $x_2$, $x_3$ and $x_4$.
Note that $\delta(G,G_1)\in\{4,6\}$.
If $F'$ is a spanning Eulerian subgraph of $G_1$,
then there exists a spanning Eulerian subgraph $F$ of $G$ with $c(F)=c(F')$ and $i(F)\le i(F')+1$,
i.e., with $\exc(F)\le\exc(F')+1$.
Hence, if $G_1$ is $2$-connected, it is a reduction of $G$.

Suppose that $G_1$ is not 2-connected.
Hence, the vertices of $G_1$ can be partitioned to non-empty sets $A$ and $B$ such that
there is at most one edge between $A$ and $B$ in $G_1$.
Since the original graph $G$ is 2-connected, both $x_1$ and $x_5$ belong to the same set, say $A$, and
the vertex $w$ to the other set, i.e., the set $B$.
In addition, at most one of the neighbors of $w$ in $G_1$ belongs to $A$ (there is at most one edge between $A$ and $B$) and
$G_1$ contains at most one of the edges $x_1x_4$ and $x_2x_5$ (for the same reason).
By symmetry, we assume that $G_1$ does not contain the edge $x_2x_5$.
If $x_1x_4$ is an edge of $G_1$, then either the edge $wx_4$ or the edge $x_1x_4$ is the edge between $A$ and $B$ and
the vertex $x_2$ must belong to $B$.
If $x_1x_4$ is not an edge of $G_1$, then at least one of the vertices $x_2$ and $x_4$ is in $B$ and
we can assume by symmetry that this vertex is $x_2$.
In either case, we have arrived at the conclusion that $x_2$ is in $B$ and $x_2x_5$ is not an edge of $G$.
Since there is at most one edge between $A$ and $B$ and the original graph is $2$-connected,
there exist disjoint paths $Q_1$ and $Q_2$,
where $Q_1$ connects the vertices $x_1$ and $x_5$ (and is fully contained in $A$) and
$Q_2$ connects the vertex $x_2$ with the vertex $x_j$ for $j=3$ or $j=4$ (and this path is fully contained in $B$).

Let $G_2$ be the graph obtained from $G-V(K)$ by adding the edge $x_2x_5$ and
a vertex $y$ that is adjacent to the vertices $x_1$, $x_3$ and $x_4$.
Note that $G_2$ is simple since $x_2x_5$ is not an edge of $G$.
If the length of $K$ in $G$ is six, we subdivide the edge $x_3y$ in addition.
Since $\delta(G,G_2)=4$ and every spanning Eulerian subgraph $F'$ of $G_2$
can be transformed to spanning Eulerian subgraph $F$ of $G$ with $\exc(F)\le \exc(F')+1$,
we get that $G_2$ is a reduction of $G$ unless $G_2$ is not $2$-connected.
We show that $G_2$ must be $2$-connected in the rest of the proof.

Suppose that $G_2$ is not $2$-connected,
i.e., the vertices of $G_2$ can be partitioned to non-empty sets $C$ and $D$ such that
there is at most one edge between $C$ and $D$ in $G_2$.
The path $Q_1$ from $x_1$ to $x_5$, the edge $x_5x_2$, the path $Q_1$ from $x_2$ to $x_j$ and
the path from $x_j$ to $x_1$ through $y$ form a cycle in $G_2$.
This implies that all the four vertices $x_1$, $x_2$, $x_j$ and $x_5$ are in the same set and
we can assume by symmetry that they are in the set $C$.
Consequently, the remaining vertex $x_{7-j}$ must be in $D$ (note that $7-j$ is either $3$ or $4$),
which implies that the edge $x_{7-j}v_{7-j}$ is a cut-edge in $G$, which is impossible.
\end{proof}

\subsection{Cycles of length six}
\label{sub-six}

Lemmas~\ref{lemma-c2e}--\ref{lemma-no5} imply that
every non-basic $2$-connected subcubic graph $G$ that is not proper has a linear-time reduction.
In this subsection, we focus on proper $2$-connected subcubic graphs that
contain a cycle of length six that satisfies some additional assumptions.
Note that such all the six vertices of such a cycle must have degree three,
each of them has a neighbor not contained in the cycle and
these neighbors are pairwise distinct.

In Lemmas~\ref{lemma-6-opp3}--\ref{lemma-6oppcuts} that we establish in this subsection,
we assume that a cycle $K$ with the properties stated in the lemmas is given.
The properties asserted by the lemmas can be checked in linear time.
In Lemma~\ref{lemma-no26}, this follows for the fact that every cycle of length six in a subcubic graph
can be intersected by at most a constant number other cycles of length six.
Since all cycles of length six can be listed in linear time (see the arguments
given at the beginning of Subsection~\ref{sub-proper}),
it is possible to find a cycle of length six with the properties given in one of the lemmas or
conclude that such a cycle does not exist in quadratic time.

\begin{lemma}
\label{lemma-6-opp3}
Let $G$ be a proper $2$-connected subcubic graph,
let $K=v_1v_2v_3v_4v_5v_6$ be a cycle of length six in $G$, and
let $x_i$ be the neighbor of $v_i$ not contained in $K$ for $i=1,\ldots, 6$.
Let $A,B$ be a partition of the vertices of $G-V(K)$ such that $x_1,x_3,x_5\in A$ and $x_2,x_4,x_6\in B$.
If $G-V(K)$ has no edge between $A$ and $B$, then $G$ has a linear-time reduction with respect to $K$.
\end{lemma}

\begin{proof}
Let $G'$ be the graph obtained from $G-V(K)$
by adding the paths $x_1z_1x_2$, $x_3z_2x_4$ and $x_5z_3x_6$,
where $z_1$, $z_2$, and $z_3$ are new vertices, each having degree two in $G'$.
Note that the graph $G'$ is simple since the vertices $x_1$, \ldots, $x_6$ are pairwise distinct as $G$ is proper.
In addition, $G'$ is $2$-connected since $G$ is $2$-connected, and $\delta(G,G')=0$.

Let $F'$ be a spanning Eulerian subgraph of $G'$.
We show that $F'$ can be transformed to a spanning Eulerian subgraph $F$ of $G$ with $\exc(F)\le \exc(F')$.
Since there are no edges between $A$ to $B$,
either one or three of the vertices $z_1$, $z_2$ and $z_3$ are isolated in $F'$.
If all of the three vertices are isolated in $F'$,
then $F$ is obtained from $F'-\{z_1,z_2,z_3\}$ by adding the cycle $K$ including its edges.
Note that $\exc(F)=\exc(F')-1$ in this case.
Suppose that only one of the vertices is isolated, say $z_3$.
Since there are no edges between $A$ and $B$,
the cycle of $F'$ that contains $z_1$ consists of the path $x_1z_1x_2$, a path from $x_2$ to $x_4$ inside $B$,
the path $x_4z_2x_3$, and a path from $x_3$ to $x_1$ inside $A$.
Let $F$ be the spanning Eulerian subgraph of $G$ obtained from $F'-\{z_1,z_2,z_3\}$
by adding the paths $x_2v_2v_3x_3$ and $x_4v_4v_5v_6v_1x_1$;
we have $c(F)=c(F')$ and $i(F)=i(F')-1$, and hence $\exc(F)=\exc(F')-1$.
We conclude that $G'$ is a reduction of $G$.
\end{proof}

Note that unlike in all the other lemmas in this section,
we consider a partition of the vertices of the original graph $G$ in the next lemma
since $G\setminus E(K)$ contains all the vertices of $G$.

\begin{lemma}\label{lemma-6-ob1}
Let $G$ be a proper $2$-connected subcubic graph and let $K=v_1\ldots v_6$ be a cycle of length six in $G$.
If there exists a partition of the vertex set of $G\setminus E(K)$ into two sets $A$ and $B$ such that
$v_1,v_3\in A$, $v_2,v_4,v_5,v_6\in B$, and there is at most one edge between $A$ and $B$,
then $G$ has a linear-time reduction with respect to $K$.
\end{lemma}

\begin{proof}
Let $G_A$ and $G_B$ be the subgraphs of $G-E(K)$ induced by $A$ and $B$, respectively.
Since $G$ is 2-connected, the graph $G_A$ is connected, and the graph $G_B$ has at most two components.

First suppose that $G_B$ is connected or has two components each containing two of the vertices $v_2$, $v_4$, $v_5$ and $v_6$.
We consider the spanning forest of $G_B$ and derive that
$G_B$ contains two disjoint paths between with the end-vertices being the neighbors of $v_2$, $v_4$, $v_5$ and $v_6$.
Let $Q_1$ be a path from $v_1$ to $v_3$ with all internal vertices in $A$, and
let $Q_2$ and $Q_3$ be the paths between two disjoint pairs of vertices $v_2$, $v_4$, $v_5$ and $v_6$ such that
their all internal vertices are in $B$.
By symmetry, we can assume that neither $Q_2$ nor $Q_3$ connects $v_5$ and $v_6$.

Suppose that $G_B$ has two components such that
one contains one and the three of the vertices $v_2$, $v_4$, $v_5$ and $v_6$.
Let $C$ be the former component.
If $C$ contains the vertex $v_5$,
then the edge between $A$ and $B$ joins a vertex of $A$ and a vertex of $C$, and
we can apply Lemma~\ref{lemma-6-opp3}.
Hence, we can assume that $C$ does not contain the vertex $v_5$, and
let $v_j$ be the vertex contained in $C$.
By symmetry, we can assume that $j\not=6$, i.e., $j=2$ or $j=4$.
Let $Q_1$ be a tree in $G$ such that its leaves are the vertices $v_1$, $v_3$ and $v_j$ and
all its vertices belong to $A$ or $C$, and
let $Q_2$ be a tree in $G$ such that its leaves are the vertices $v_{6-j}$, $v_5$ and $v_6$ (note that $6-j$ is $2$ or $4$) and
all its vertices belong to the component of $G_B$ different from $C$.

Let $G'$ be the graph obtained from $G\setminus E(K)$ by identifying the vertices $v_1$ and $v_5$ to a single vertex $z_{15}$,
identifying $v_2$ and $v_4$ to a single vertex $z_{24}$, and identifying $v_3$ and $v_6$ to a single vertex $z_{36}$.
The paths and trees $Q_i$, which we have constructed in the previous two paragraphs,
yield that the graph $G'$ is $2$-connected.
Note that $\delta(G,G')=0$.

We establish that $G'$ is a reduction of $G$.
Let $F'$ be a spanning Eulerian subgraph of $G'$.
If all three vertices $z_{15}$, $z_{24}$ and $z_{36}$ are isolated in $F'$,
then we can extend $F$ by adding a cycle $K$ to an Eulerian spanning subgraph $F$ of $G$ with $\exc(F)=\exc(F')-1$.
If two of the vertices $z_{15}$, $z_{24}$ and $z_{36}$ are isolated in $F'$,
then we can extend by rerouting one of the cycles of $F'$ through the cycle $K$
to an Eulerian spanning subgraph $F$ of $G$ with $\exc(F)\in\{\exc(F')-1,\exc(F')\}$.
Finally, if one of the vertices $z_{15}$, $z_{24}$ and $z_{36}$ is an isolated vertex in $F'$,
it is possible to reroute the cycle(s) of $F'$ containing the two of the vertices $z_{15}$, $z_{24}$ and $z_{36}$
to get an Eulerian spanning subgraph $F$ such that
the number of non-trivial components of $F$ does not exceed that of $F'$ and
the same is true for the number of isolated vertices, i.e., $\exc(F)\le\exc(F')$.
Hence, none of the vertices $z_{15}$, $z_{24}$ and $z_{36}$ isolated in $F'$.

If the vertices $z_{15}$, $z_{24}$ and $z_{36}$ are contained in at least two different cycles of $F'$,
it is possible to complete the three paths of $F'-\{z_{15},z_{24},z_{36}\}$
to an Eulerian spanning subgraph $F$ of $G$ in a way that
there are at most two cycles of $F$ passing through the cycle $K$ and
none of the vertices of $K$ is isolated in $F$. In particular, $\exc(F)\le\exc(F')$.
Consequently, we can assume that all the vertices $z_{15}$, $z_{24}$ and $z_{36}$ are contained in the same cycle of $F'$.
Let $R$, $R'$ and $R''$ be the paths of this cycle after removing the vertices $z_{15}$, $z_{24}$ and $z_{36}$.

Observe that one of the paths is fully contained in $A$ and connects the neighbors of the vertices $v_1$ and $v_3$;
let $R$ be this path.
Since the paths $R$, $R'$ and $R''$ together with the vertices $z_{15}$, $z_{24}$ and $z_{36}$ form a cycle,
it follows that neither the path $R'$ nor the path $R''$ connects the neighbors of the vertices $v_5$ and $v_6$.
Hence, $G$ contains a cycle formed by the paths $R$, $R'$, $R''$,
the edges joining $v_1,\ldots,v_6$ to their numbers outside $K$ and
the edges $v_1v_2$, $v_3v_4$ and $v_5v_6$.
Replacing the cycle of $F'$ containing the vertices $z_{15}$, $z_{24}$ and $z_{36}$ with this cycle
yields an Eulerian spanning subgraph $F$ of $G$ with $\exc(F)=\exc(F')$.
This finishes the proof that $G'$ is a reduction of $G$.
\end{proof}

In the next two lemmas,
we show that two different types of cycles of length six that are not $\theta$-cycles can be reduced.

\begin{lemma}\label{lemma-6-ob0}
Let $G$ be a proper $2$-connected subcubic graph,
let $K=v_1v_2v_3v_4v_5v_6$ be one of its cycles of length six, and
let $x_i$ be the neighbor of $v_i$ not contained in $K$ for $i=1,\ldots,6$.
Let $A,B$ be a partition of the vertices of $G-V(K)$ such that $x_1,x_2\in A$, $x_3,x_4,x_5,x_6\in B$, and
there is no edge between $A$ and $B$.
If $K$ is not a $\theta$-cycle, then $G$ has a linear-time reduction with respect to $K$.
\end{lemma}

\begin{proof}
Since $G$ is proper, the vertices $x_1, \ldots, x_6$ are pairwise distinct.
Let $G_A$ and $G_B$ be the subgraphs of $G$ induced by $A$ and $B$.
The $2$-connectivity of $G$ implies that $G_A$ is connected and $G_B$ has at most two components,
each containing two vertices among $x_3,\ldots,x_6$.
If $G_B$ contains an edge-cut of size at most one separating $\{v_3,v_5\}$ from $\{v_4,v_6\}$,
then a reduction of $G$ can be obtained using Lemma~\ref{lemma-6-ob1},
which we apply with one of the sides of this cut in $G_B$ playing the role of $A$ and
the rest of the vertices outside the cycle $B$ playing the role of $B$ in the statement of Lemma~\ref{lemma-6-ob1}.
We conclude that $G-V(K)$ contains three disjoint paths $Q_1$, $Q_2$, and $Q_3$ such that
$Q_1$ connects $x_1$ with $x_2$, $Q_2$ connects $x_3$ with $x_4$ or $x_6$, and
$Q_3$ connects $x_5$ with the other of the vertices $x_4$ and $x_6$.

Let $G_1$ be the graph obtained from $G-V(K)$ by adding paths $x_1z_1x_4$, $x_2z_2x_5$, and $x_3z_3x_6$,
where $z_1$, $z_2$ and $z_3$ are new vertices, each having degree two in $G_1$.
Note that $\delta(G,G_1)=0$.
We show that $G_1$ is a reduction of $G$ assuming that $G_1$ is $2$-connected.
Let $F_1$ be a spanning Eulerian subgraph of $G_1$.
If at least two of the vertices $z_1$, $z_2$ and $z_3$ are isolated in $F_1$,
then it is easy to construct a spanning Eulerian subgraph $F$ of $G$ with $\exc(F)\le\exc(F_1)$.
Hence, assume that at most one of the vertices $z_1$, $z_2$ and $z_3$ is isolated in $F_1$.
Since $z_1$ and $z_2$ is a $2$-vertex cut in $G_1$, it follows that
the paths $x_1z_1x_4$ and $x_2z_2x_5$ are contained in the same cycle of $F_1$.
If $z_3$ is isolated in $F_1$,
then let $F$ be a spanning Eulerian subgraph of $G$ obtained from $F_1-\{z_1,z_2,z_3\}$
by adding the paths $x_1v_1v_6v_5x_5$ and $x_2v_2v_3v_4x_4$.
Note that $\exc(F)=\exc(F_1)-1$ in this case.
If $z_3$ is not isolated in $F_1$, i.e., the path $x_3z_3x_6$ is contained in a cycle of $F_1$,
then let $F$ be obtained from $F_1-\{z_1,z_2,z_3\}$
by adding the paths $x_1v_1v_6x_6$, $x_2v_2v_3x_3$ and $x_4v_4v_5x_5$.
Observe that $c(F)=c(F_1)$, which implies $\exc(F)=\exc(F_1)$.
We conclude that $G_1$ is a reduction of $G$ if $G_1$ is $2$-connected.

It remains to consider the case that $G_1$ is not 2-connected.
This implies that the path $Q_2$ connects $x_3$ with $x_6$, and
the path $Q_3$ connects $x_4$ with $x_5$.
In addition, the vertices of the subgraph $G_B$ can be split into two parts $B'$ and $B''$ such that
$B'$ contains the vertices $x_3$ and $x_6$, $B''$ contains the vertices $x_4$ and $x_5$, and
there is at most one edge between $B'$ and $B''$.
Since $K$ is not a $\theta$-cycle, there must be at least one edge between $B'$ and $B''$,
i.e., there is exactly one edge between $B'$ and $B''$.
Let $e$ be this edge.

Let $G_2$ be the graph obtained from $G-V(K)$ by adding the edges $x_2x_3$, $x_1x_4$ and $x_5x_6$, and
by subdividing $e$ by one new vertex $w$.
Observe that $G_2$ is $2$-connected and $\delta(G,G_2)=4$.
In addition, $G_2$ is simple since $G$ is proper.
We show that $G_2$ is a reduction of $G$.
Let $F_2$ be a spanning Eulerian subgraph of $G_2$.
If $w$ is an isolated vertex in $F_2$,
then $F'$ contains either none or all of the edges $x_2x_3$, $x_1x_4$ and $x_5x_6$.
In the former case, let $F$ be the spanning Eulerian subgraph of $G$ obtained from $F'$ by adding the cycle $K$.
In the latter case, let $F$ be the subgraph obtained from $F'$ by removing the edges $x_2x_3$, $x_1x_4$ and $x_5x_6$ and
adding the paths $x_2v_2v_3x_3$, $x_4v_4v_5x_5$ and $x_6v_6v_1x_1$.
Since $c(F)=c(F_2)+1$ and $i(F)=i(F_2)-1$ in either of the cases, it follows that $\exc(F)=\exc(F_2)+1$.

If $w$ is not an isolated vertex,
then the subgraph $F_2$ either contains the edge $x_5x_6$ or
it contains the edges $x_2x_3$ and $x_1x_4$.
In the former case, let $F$ be the spanning Eulerian subgraph of $G$ obtained from $F'$
by removing the edge $x_5x_6$ and adding the path $x_6v_6v_1\cdots v_5x_5$.
In the latter case, let $F$ be the spanning Eulerian subgraph of $G$ obtained from $F'$
by removing the edges $x_2x_3$ and $x_1x_4$, and
adding the paths $x_2v_2v_3x_3$ and $x_1v_1v_6v_5v_4x_4$.
In both case, we get that $c(F)=c(F_2)$ and $i(F)=i(F_2)$, which yields that $\exc(F)=\exc(F_2)$.
This concludes the proof that $G_2$ is a reduction of $G$.
\end{proof}

\begin{lemma}\label{lemma-6-oppa}
Let $G$ be a proper $2$-connected subcubic graph,
let $K=v_1v_2v_3v_4v_5v_6$ be one of its cycles of length six, and
let $x_i$ be the neighbor of $v_i$ not contained in $K$ for $i=1,\ldots,6$.
If $K$ is not a $\theta$-cycle and
the vertices $x_1$ and $x_4$ are in different components of $G-V(K)$,
then $G$ has a linear-time reduction with respect to $K$.
\end{lemma}

\begin{proof}
Let $A$ and $B$ a partition of the vertices of $G-V(K)$ such that
$x_1\in A$ and $x_4\in B$, and there is no edge between $A$ and $B$.
By symmetry, we can assume that $|A\cap \{x_1,\ldots, x_6\}|\le 3$.
If $x_3\in A$ or $x_5\in A$,
then the reduction exists by Lemma~\ref{lemma-6-ob1};
e.g., if $x_3\in A$, apply the lemma with $A\cup\{v_1,v_3\}$ playing the role of the set $A$ and
with $B\cup\{v_2,v_4,v_5,v_6\}$ playing the role of the set $B$ from the statement of the lemma.
If $A=\{x_1,x_2,x_6\}$,
then the reduction also exists by Lemma~\ref{lemma-6-ob1}:
apply the lemma with $A\cup\{v_6,v_2\}$ playing the role of the set $A$ and
with $B\cup\{v_1,v_3,v_4,v_5\}$ playing the role of the set $B$.
We conclude that $A\subseteq\{x_1,x_2,x_6\}$ and $|A|=2$.
By symmetry, we can assume that $A=\{x_1,x_2\}$ and $B=\{x_3,x_4,x_5,x_6\}$.
The existence of the reduction now follows from Lemma~\ref{lemma-6-ob0}.
\end{proof}

Lemmas~\ref{lemma-6-ob0} and~\ref{lemma-6-oppa} yield the following.

\begin{lemma}\label{lemma-6-nocut}
Let $G$ be a proper $2$-connected subcubic graph.
If $G$ contains a cycle $K$ of length six that is not a $\theta$-cycle and that contains an edge in $2$-edge-cut,
then $G$ has a linear-time reduction with respect to $K$.
\end{lemma}

\begin{proof}
Let $v_1,\ldots,v_6$ be the vertices of the cycle $K$.
By symmetry, we can assume that the edge $v_1v_2$ is contained in a $2$-edge-cut.
The $2$-edge-cut must contain another edge $e$ of the cycle $K$.
Since $G$ is $2$-connected, this edge $e$ is neither $v_1v_6$ nor $v_2v_3$.
If the edge $e$ is $v_3v_4$ or $v_5v_6$, then the reduction exists by Lemma~\ref{lemma-6-ob0}.
Otherwise, the edge $e$ is the edge $v_4v_5$ and the reduction exists by Lemma~\ref{lemma-6-oppa}.
\end{proof}

\begin{lemma}\label{lemma-6-mainred}
Let $G$ be a proper $2$-connected subcubic graph,
let $K=v_1v_2v_3v_4v_5v_6$ be one of its cycles of length six, and
let $x_i$ be the neighbor of $v_i$ not contained in $K$ for $i=1,\ldots,6$.
If the edge $v_1x_1$ is not contained in a $2$-edge-cut, then $G$ has a linear-time reduction unless
\begin{itemize}
\item all the edges $v_2x_2$, $v_3x_3$, $v_5x_5$, and $v_6x_6$ are contained in $2$-edge-cuts, and
\item there exists a partition $A$ and $B$ of the vertices of $G-V(K)$ such that
      $x_1,x_2,x_6\in A$, $x_3,x_4,x_5\in B$,
      $G-V(K)$ contains exactly one edge between $A$ and $B$, and
      both the subgraphs induced by $A$ and $B$ are connected.
\end{itemize}
\end{lemma}

\begin{proof}
The cycle $K$ is not a $\theta$-cycle since all edges incident with a $\theta$-cycle are contained in $2$-edge-cuts.
Since the edge $v_1x_1$ is not contained in a $2$-edge-cut, the degree of $x_1$ is three,
in particular, its degree in $G-V(K)$ is two.
Note that $G-V(K)$ contains a path $Q_{25}$ connecting the vertex $x_2$ with the vertex $x_5$,
a path $Q_{36}$ connecting $x_3$ with $x_6$, and a path $Q_{14}$ connecting $x_1$ with $x_4$ (the three paths need not be disjoint)
since otherwise the existence of the reduction of $G$ follows from Lemma~\ref{lemma-6-oppa}.

Let $G_1$ be the graph obtained from $G-V(K)$ by adding the edge $x_2x_6$ and a vertex $z$ adjacent to $x_3$, $x_4$, and $x_5$.
Note that $\delta(G,G_1)=4$.
Since $G$ is proper, the vertices $x_2$ and $x_6$ are not adjacent in $G$.
Hence $G_1$ is a simple subcubic graph.
Observe that any spanning Eulerian subgraph $F_1$ of $G_1$
can be transformed to a spanning Eulerian subgraph $F$ of $G$ with $\exc(F)\le \exc(F_1)+1$.
Hence, $G_1$ is a reduction of $G$ unless $G_1$ is not 2-connected.

In the rest of the proof, we assume that $G_1$ is not 2-connected. 
This implies that there exists a partition of vertices of $G_1$ to non-empty sets $A$ and $B$
such that there is at most one edge between $A$ and $B$.
By symmetry, we can assume that $x_2$ is contained in $A$.
Note that the edge $x_2x_6$ and the paths $Q_{36}$, $x_3zx_5$ and $Q_{25}$
contain a cycle passing through the edge $x_2x_6$ and a cycle passing through the path $x_3zx_5$;
note that their union need not be a cycle since the path $Q_{36}$ and $Q_{25}$ need not be disjoint.
This implies that $x_6\in A$, and either $\{x_3,x_5\}\subseteq A$ or $\{x_3,x_5\}\subseteq B$.
If $\{x_3,x_5\}\subseteq A$, then either $G$ is be 2-connected (if $x_1\in A$),
or the edge $v_1x_1$ is contained in a $2$-edge-cut in $G$ (if $x_1\in B$).
Since both these conclusions are impossible,we get that $\{x_3,x_5\}\subseteq B$.
Hence, there is an edge between $A$ and $B$ and this edge is contained in both paths $Q_{25}$ and $Q_{36}$.
Let $e_0$ be this edge.
Observe that $e_0$ is not incident with the vertex $z$, which does not exist in $G$.
In particular, both the vertices $z$ and $x_4$ belong to $B$.
If $x_1\in B$, then Lemma~\ref{lemma-6-ob1} yields the existence of a reduction of $G$.
So, we can assume that $x_1\in A$.
This yields that the path $Q_{14}$ also contains the edge $e_0$.

Since all paths $Q_{14}$, $Q_{25}$, and $Q_{36}$ must contain the edge $e_0$,
we conclude that $G-V(K)-e_0$ has exactly two components;
one of the two components has the vertex set $A$, in particular, it contains the vertices $x_1$, $x_2$ and $x_6$, and
the other component has the vertex set $B\setminus\{z\}$ and contains the vertices $x_3$, $x_4$ and $x_5$.
If $v_ix_i$ is not contained in a $2$-edge-cut for some $i\in\{2,3,5,6\}$, say $i=2$, 
then consider the graph $G_2$ obtained from $G-V(K)$ by adding the edge $x_1x_3$ and
a new vertex $z$ adjacent to $x_4$, $x_5$, and $x_6$.
If the graph $G_2$ were not $2$-connected, it is easy to see that $G$ would not be $2$-connected.
Hence, $G_2$ is a reduction for $G$ (note that the edge $v_ix_i$ can play the role of the edge $v_1x_1$
at the beginning of our proof).
\end{proof}

\begin{lemma}\label{lemma-6no2e}
Let $G$ be a proper $2$-connected subcubic graph,
let $K=v_1v_2v_3v_4v_5v_6$ be one of its cycles of length six, and
let $x_i$ be the neighbor of $v_i$ not contained in $K$ for $i=1,\ldots,6$.
If neither $v_1x_1$ nor $v_4x_4$ is contained in a $2$-edge-cut, then $G$ has a linear-time reduction with respect to $K$.
\end{lemma}

\begin{proof}
Lemma~\ref{lemma-6-mainred} yields that there either exists a reduction of $G$ or
a partition $A$ and $B$ of the vertices of $G-V(K)$ such that $x_1,x_2,x_6\in A$ and $x_3,x_4,x_5\in B$,
$G-V(K)$ contains exactly one edge $e$ between $A$ and $B$, and
both subgraphs of $G-V(K)$ induced by $A$ and $B$ are connected.
In the former case, the proof of the lemma is finished.
So, we focus on the latter case.

Let $G'$ be the graph obtained from $G-V(K)$ by adding the edges $x_2x_3$ and $x_5x_6$, and by subdividing the edge $e$ twice.
Observe that $G'$ is a 2-connected simple subcubic graph and $\delta(G,G_1)=0$.
Let $F'$ be a spanning Eulerian subgraph of $G'$.
If $F'$ does not contain the path corresponding to the edge $e$ or one of the edges $x_2x_3$ and $x_5x_6$,
then $F'$ does not contain any of the edges $x_2x_3$ and $x_5x_6$, and 
$G$ has a spanning Eulerian subgraph $F$ such that $c(F)=c(F')+1$ and $i(F)=i(F')-2$, i.e., $\exc(F)=\exc(F')$.
If $F'$ does not contain the path corresponding to the edge $e$ but contains both the edges $x_2x_3$ and $x_5x_6$,
then $G$ has a spanning Eulerian subgraph $F$ such that $c(F)=c(F')$ and $i(F)=i(F')$, i.e., $\exc(F)=\exc(F')$.
Finally, if $F'$ contains the path corresponding to the edge $e$,
then $F'$ contains one of the edges $x_2x_3$ and $x_5x_6$, and
$G$ has a spanning Eulerian subgraph $F$ such that $c(F)=c(F')$ and $i(F)=i(F')$, i.e., $\exc(F)=\exc(F')$.
In all the case, it holds that $G$ has a spanning Eulerian subgraph $F$ with $\exc(F)\le \exc(F_1)$.
We conclude that $G'$ is a reduction of $G$.
\end{proof}

We now combine Lemmas~\ref{lemma-6-nocut}, \ref{lemma-6-mainred} and \ref{lemma-6no2e}.

\begin{lemma}\label{lemma-no26}
Let $G$ be a proper $2$-connected subcubic graph and
let $K$ and $K'$ be two distinct cycles of length six in $G$.
If the cycles $K$ and $K'$ intersect and
at least one of them is not a $\theta$-cycle,
then $G$ has a linear-time reduction with respect to $K\cup K'$.
\end{lemma}

\begin{proof}
We can assume that $K$ is not a $\theta$-cycle by symmetry.
Since the cycles $K$ and $K'$ are distinct,
the cycle $K$ is incident with at least two edges of $K'$ not contained in $K$.
None of these edges is contained in a $2$-edge-cut by Lemma~\ref{lemma-6-nocut}.
By Lemma~\ref{lemma-6-mainred}, these two edges must be incident with the opposite vertices of $K$.
Finally, the existence of the reduction follows from Lemma~\ref{lemma-6no2e}.
\end{proof}

We finish this subsection with two additional lemmas on edges incident with cycles of length six that
are contained in $2$-edge-cuts.

\begin{lemma}\label{lemma-6adjcuts}
Let $G$ be a proper $2$-connected subcubic graph,
let $K=v_1v_2v_3v_4v_5v_6$ be one of its cycles of length six, and
let $x_i$ be the neighbor of $v_i$ not contained in $K$ for $i=1,\ldots,6$.
If $K$ is not a $\theta$-cycle, and
there exists $i<j$ such that $j-i\neq 3$ and $\{v_ix_i, v_jx_j\}$ is a $2$-edge-cut,
then $G$ has a linear-time reduction with respect to $K$.
\end{lemma}

\begin{proof}
By symmetry, we can assume that $j-i$ is equal to $1$ or $2$.
If $j-i=1$, then the existence of the reduction follows from Lemma~\ref{lemma-6-nocut}, and
if $j-i=2$, then its existence follows from Lemma~\ref{lemma-6-ob1}.
\end{proof}

\begin{lemma}\label{lemma-6oppcuts}
Let $G$ be a proper $2$-connected subcubic graph,
let $K=v_1v_2v_3v_4v_5v_6$ be one of its cycles of length six, and
let $x_i$ be the neighbor of $v_i$ not contained in $K$ for $i=1,\ldots,6$.
If there exists $1\le i<j\le 3$ such that
both $\{v_ix_i,v_{i+3}x_{i+3}\}$ and $\{v_jx_j,v_{j+3}x_{j+3}\}$ are $2$-edge-cuts in $G$,
then $G$ has a linear-time reduction with respect to $K$.
\end{lemma}

\begin{proof}
Sine $G$ is $2$-connected,
$G-V(K)$ has three components $C_1$, $C_2$, and $C_3$, and
the vertices $x_i$ and $x_{i+3}$ are contained in $C_i$ for $i\in\{1,2,3\}$.
Let $G'$ be the graph obtained from $G-V(K)$ by adding the edges $x_1x_5$ and $x_2x_6$, and
the path $x_3wx_4$, where $w$ is a new vertex, which have degree two in $G'$.
Note that $G'$ is a simple $2$-connected subcubic graph and $\delta(G,G')=4$.
Let $F'$ be a spanning Eulerian subgraph of $G'$.
The subgraph $F'$ either contains none of the edges $x_1x_5$, $x_2x_6$, $x_3w$ and $x_4w$, or
it contains all of them.
In the former case, $G$ has a spanning Eulerian subgraph $F$ with $c(F)=c(F')+1$ and $i(F)=i(F')-1$, i.e., $\exc(F)=\exc(F')+1$.
In the latter case, $G$ has a spanning Eulerian subgraph $F$ with $c(F)=c(F')$ and $i(F)=i(F')$, i.e., $\exc(F)=\exc(F')$.
It follows that $G'$ is a reduction of $G$.
\end{proof}

\subsection{Cycles of length seven}
\label{sub-seven}

In this subsection, we establish two lemmas concerning the reductions involving cycles of length seven.
As in Subsection~\ref{sub-six}, we assume that a cycle $K$ with the properties stated in the lemmas is given.
Since the properties asserted by Lemmas~\ref{lemma-no2in7} and~\ref{lemma-cutsin7} can be checked in linear time and
all cycles of length seven can be listed in linear time, it is possible to find a cycle of length seven
with the properties given in one of Lemmas~\ref{lemma-no2in7} and~\ref{lemma-cutsin7} or
conclude that such a cycle does not exist in quadratic time.

To prove the first of the lemmas, we need to use the Splitting Lemma of Fleischner~\cite{fleischsplit},
which we now state.
Let us introduce some additional notation.
We say that the graph $G'$ is obtained from a graph $G$ by \emph{splitting off} the edges $u_1v$ and $u_2v$
if the graph is obtained by removing the edges $u_1v$ and $u_2v$ and adding the edge $u_1u_2$.
We will always apply this operation to edges incident with the same vertex.
We can now state the Splitting Lemma.

\begin{lemma}[Splitting Lemma]\label{lemma-fleisch}
Let $G$ be a $2$-edge-connected graph and let $v$ be a vertex of degree at least $4$.
\begin{itemize}
\item If $v$ is a cut-vertex and $e_1$ and $e_2$ are two edges incident with $v$ that belong to different blocks of $G$,
      then splitting off $e_1$ and $e_2$ results in a 2-edge-connected graph.
\item If $v$ is not a cut-vertex and $e_1$, $e_2$, and $e_3$ are edges incident with $v$,
      then splitting off $e_1$ and $e_2$ or splitting off $e_2$ and $e_3$ results in a 2-edge-connected graph.
\end{itemize}
\end{lemma}

We are now ready to prove the first lemma of this subsection.

\begin{lemma}\label{lemma-no2in7}
Let $G$ be a proper $2$-connected subcubic graph.
If $G$ has a cycle $K$ of length seven that contains a vertex of degree two,
then $G$ has a linear-time reduction with respect to $K$.
\end{lemma}

\begin{proof}
Since $G$ is proper, $K$ is an induced cycle and at most two vertices of $K$ have degree two.
Let $v_1, \ldots, v_k$ be the vertices of $K$ of degree three in order around the cycle;
note that $k$ is five or six.
Further, let $x_i$ be the neighbor of $v_i$ outside of $K$ for $i\in\{1,\ldots,k\}$.
Since $G$ is proper,
the vertices $x_1,\ldots,x_k$ are pairwise distinct.
Moreover,
if $i\neq j$ and either $|i-j|\le 2$, or $|i-j|\ge k-2$, then $x_ix_j$ is not an edge of $G$.
Let $G'$ be the graph obtained from $G$ by contracting the cycle $K$ to a single vertex $w$.
By Lemma~\ref{lemma-fleisch} and symmetry, we can assume that
the graph $G''$ obtained from $G'$ by splitting off $wx_1$ and $wx_2$ is $2$-edge-connected.

We first deal with the case that $k=5$.
Note that $G''$ is a simple $2$-connected subcubic graph and $\delta(G,G'')=8$.
Let $F'$ be a spanning Eulerian subgraph of $G''$.
If the vertex $w$ is isolated in $F'$ and $F'$ does not contain the edge $x_1x_2$,
then there exists a spanning Eulerian subgraph $F$ of $G$ with $c(F)=c(F')+1$ and $i(F)=i(F')-1$.
If either the vertex $w$ is isolated and $F'$ contains the edge $x_1x_2$, or
the vertex $w$ is not isolated and $F'$ does not contain the edge $x_1x_2$,
then there exists a spanning Eulerian subgraph $F$ of $G$ with $c(F)=c(F')$ and $i(F)\le i(F')+1$.
Finally, if the vertex $w$ is not isolated and $F'$ contains the edge $x_1x_2$,
then there exists a spanning Eulerian subgraph $F$ of $G$ with $c(F)=c(F')$ and $i(F)\le i(F')+3$, and
if there is no such subgraph $F$ with $i(F)\le i(F')+2$,
then there is also a spanning Eulerian subgraph $F$ with $c(F)\le c(F')+1$ and $i(F)=i(F')$.
In all the cases, we conclude that there is a spanning Eulerian subgraph $F$ of $G$ with $\exc(F)\le\exc(F')+2$.

We now deal with the case $k=6$.
If $G'$ is not $2$-connected, then $w$ must be a cut-vertex and $G'$ has two blocks,
each containing two neighbors of $w$.
Regardless whether $w$ is a cut-vertex,
Lemma~\ref{lemma-fleisch} implies that splitting off $wx_3$ with either $wx_4$ or $wx_5$ and
suppressing $w$ yields a $2$-connected subcubic graph $G''$.
Note that $G''$ is simple and $\delta(G,G'')=8$.

Let $e$, $e'$ and $e''$ be the edges of $G''$ not contained in $G$.
Consider a spanning Eulerian subgraph $F'$ of $G''$.
If $F'$ uses none of the edges $e$, $e'$ and $e''$ ,
then there exists a spanning Eulerian subgraph $F$ of $G$ with $c(F)=c(F')+1$ and $i(F)=i(F')$.
If $F'$ uses exactly one of the edges $e$, $e'$ and $e''$,
then there exists a spanning Eulerian subgraph $F$ with $c(F)=c(F')$ and $i(F)\le i(F')+2$.
If $F'$ uses exactly two of the edges $e$, $e'$ and $e''$,
then there exists a spanning Eulerian subgraph $F$ with $c(F)\le c(F')+1$ and $i(F)\le i(F')+1$, and
if there is no such subgraph $F$ with $c(F)<c(F')+1$ or $i(F)<i(F')+1$,
then there is also a spanning Eulerian subgraph $F$ with $c(F)=c(F')$ and $i(F)=i(F')+2$.
Finally, if $F'$ uses all the edges $e$, $e'$ and $e''$,
then there exists a spanning Eulerian subgraph $F$ with $c(F)\le c(F')+2$ and $i(F)=i(F')$, and
if there is no such subgraph $F$ with $c(F)\le c(F')+1$,
then there is also spanning Eulerian subgraph $F$ with $c(F)=c(F')$ and $i(F)=i(F')+1$.
In all the cases, there exists a spanning Eulerian subgraph $F$ of $G$ with $\exc(F)\le \exc(F')+2$,
i.e., $G''$ is a reduction of $G$.
\end{proof}

Note that Lemmas~\ref{lemma-no5} and~\ref{lemma-no2in7} yield that
if a proper $2$-connected subcubic graph $G$ contains a cycle of length at most seven that contains a vertex of degree two,
then $G$ has a linear-time reduction.
We next prove the final lemma of this section.

\begin{lemma}\label{lemma-cutsin7}
Let $G$ be a proper $2$-connected subcubic graph and let $K=v_1v_2\ldots v_m$ be a cycle in $G$ of length at most $7$.
If each of the edges $v_1v_m$ and $v_2v_3$ is contained in a $2$-edge-cut
but the edges $v_1v_m$ and $v_2v_3$ themselves do not form a $2$-edge-cut,
then $G$ has a linear-time reduction with respect to $K$.
\end{lemma}

\begin{proof}
Since $G$ is proper, the length of $K$ is at least six, i.e., $m\ge 6$.
If $m=6$, then the existence of a reduction of $G$ follows from Lemma~\ref{lemma-6-nocut}.
Hence, we can assume that $m=7$.
In addition, all the vertices of $K$ have degree three (otherwise, Lemma~\ref{lemma-no2in7} yields the existence of a reduction).
For $i=1,\ldots,7$, let $x_i$ be the neighbor of the vertex $v_i$ outside of $K$.
Let $e_{17}$ be an edge forming a $2$-edge-cut with the edge $v_1v_7$, and
let $e_{23}$ be an edge forming a $2$-edge-cut with the edge $v_2v_3$.
Note that the edges $e_{17}$ and $e_{23}$ must be edges of the cycle $K$.
Moreover, since $G$ is $2$-connected and the edges $v_1v_7$ and $v_2v_3$ do not form a $2$-edge-cut,
it follows that $e_{17}$ is one of the edges $v_3v_4$, $v_4v_5$ and $v_5v_6$ and
$e_{23}$ is one of the edges $v_4v_5$, $v_5v_6$ and $v_6v_7$.

For $(i,j)\in\{(1,7),\,(2,3)\}$,
let $U_{ij}$ and $V_{ij}$ be the two sides of the $2$-edge-cut formed by the edges $v_iv_j$ and $e_{ij}$;
by symmetry, we can assume that $v_i\in U_{ij}$.
Since neither $e_{17}$ nor $e_{23}$ is the edge $v_1v_2$, we have $v_1,v_2\in U_{17}\cap U_{23}$.
Since $e_{17}\neq v_2v_3$ and $e_{23}\neq v_1v_7$, we have $v_3\in U_{17}\cap V_{23}$ and $v_7\in U_{23}\cap V_{17}$.
Finally, since $e_{23}\neq v_3v_4$, we have $v_4\in V_{23}$, and the symmetric arguments yields that $v_6\in V_{17}$.

Suppose that the vertex $v_5$ is contained in $V_{17}$.
If $v_4$ were also contained in $V_{17}$, then the edge $v_3x_3$ would be a cut-edge in $G$, which is impossible.
If $v_4$ were not contained in $V_{17}$, i.e., it were contained in $U_{17}$,
then the edges $v_1v_7$ and $v_2v_3$ would form a $2$-edge-cut, which is also impossible.
Hence, the vertex $v_5$ must be contained in $U_{17}$.
The symmetric argument yields that $v_5$ is contained in $U_{23}$.
Consequently, the edge $e_{17}$ is the edge $v_5v_6$ and the edge $e_{23}$ is the edge $v_4v_5$.
It follows that $G-V(K)$ has three components with vertex sets
$A=(U_{17}\cap U_{23})\setminus V(K)$, $B=(U_{17}\cap V_{23})\setminus V(K)$, and $C=(U_{23}\cap V_{17})\setminus V(K)$.
Note that $x_1,x_2,x_5\in V(A)$, $x_3,x_4\in V(B)$, and $x_6,x_7\in V(C)$.

Let $G'$ be obtained from $G$ by removing $v_1$ and $v_2$,
adding edges $v_3x_1$ and $v_7x_2$, and by subdividing the edge $v_5x_5$ once;
let $w$ be the new vertex of degree two.
The graph $G'$ is a simple $2$-connected subcubic graph and $\delta(G,G')=0$.
Let $F'$ be a spanning Eulerian subgraph of $G'$.
Let $F''$ be the spanning Eulerian subgraph of $G$ obtained from $F'$ as follows.
First, include the vertices $v_1$ and $v_2$ as isolated vertices.
If $F'$ contains the path $v_5wx_5$, then replace it with the edge $v_5x_5$; otherwise, remove the vertex $w$.
If $F'$ contains the edge $v_3x_1$, include the edges $x_1v_1$ and $v_3v_2$, and
if $F'$ contains the edge $v_7x_2$, include the edges $x_2v_2$ and $v_7v_1$.
Finally, include the edge $v_1v_2$ if the vertices $v_1$ and $v_2$ have odd degree so far.
Note that the resulting graph $F''$ is a spanning Eulerian subgraph of $G$,
$c(F'')=c(F')$, $i(F'')\le i(F')+1$, and if $i(F'')=i(F')+1$,
then all the three vertices $v_1$, $v_2$ and $v_5$ are isolated in $F''$.
If $i(F'')\le i(F')$, set $F=F''$; otherwise, set $F$ the be the spanning Eulerian subgraph of $G$
with the edge set equal to the symmetric difference of $E(F)$ and $E(K)$.
In the latter case, $c(F)\le c(F')$ and $i(F)=i(F')-2$.
It follows that $G'$ is a reduction of $G$.
\end{proof}

\subsection{Clean subcubic graphs}
\label{sub-clean}

We now summarize the facts that have been established in this section.
We will call a non-basic $2$-connected subcubic graph $G$ clean
if none of the lemmas that we have proven can be applied to $G$.
Formally, a $2$-connected subcubic graph $G$ is \emph{clean} if it is proper and
\begin{itemize}
\item[(CT1)] no cycle of length at most $7$ in $G$ contains a vertex of degree two,
\item[(CT2)] every cycle of length six in $G$ that is not a $\theta$-cycle is disjoint from all other cycles of length six,
\item[(CT3)] every cycle $K=v_1\ldots v_m$ of length $m\le 7$ in $G$ satisfies that
             if each of the edges $v_1v_m$ and $v_2v_3$ is contained in a $2$-edge-cut,
             then the edges $v_1v_m$ and $v_2v_3$ themselves form a $2$-edge-cut, and
\item[(CT4)] every cycle $K=v_1\ldots v_6$ of length six in $G$ satisfies at least one of the following
\begin{itemize}
\item[(a)] $K$ is a $\theta$-cycle, or
\item[(b)] each edge exiting $K$ is contained in a $2$-edge-cut but no two of them together form a $2$-edge-cut, or
\item[(c)] each edge exiting $K$ is contained in a $2$-edge-cut, and there exists exactly one pair $i$ and $j$ with $1\le i<j\le 6$
           such that the edges $v_ix_i$ and $v_jx_j$ form a $2$-edge-cut, and this pair satisfies $j-i=3$, or,
\item[(d)] precisely one edge exiting $K$, say $v_1x_1$, is not contained in a $2$-edge-cut, and
           there exists a partition $A$ and $B$ of the vertices of $G-V(K)$
	   such that $x_1,x_2,x_6\in A$, $x_3,x_4,x_5\in B$,
	   there is exactly one edge between $A$ and $B$, and
	   both $A$ and $B$ induce connected subgraphs of $G-V(K)$,
\end{itemize}
           where $x_i$ is the neighbor of the vertex $v_i$ outside the cycle $K$, $i\in\{1,\ldots,6,\}$.
\end{itemize}

Summarizing the results of this section, we get the following.

\begin{theorem}
\label{cor-summary}
There exists an algorithm running in time $O(n^3)$ that
constructs for a given $n$-vertex $2$-connected subcubic graph $G$
reduction of $G$ that is either basic or clean.
\end{theorem}

\begin{proof}
We show that if $G$ is neither basic nor clean,
then one of Lemmas~\ref{lemma-c2e}--\ref{lemma-cutsin7} applies.
As discussed in Subsections~\ref{sub-proper}--\ref{sub-seven},
it is possible to check the existence of a reduction as described in these lemmas,
to find the corresponding subgraph and to perform the reduction in quadratic time.
Since each step results in decreasing the sum $n(G)+n_2(G)$,
the algorithm stops after at most $O(n)$ steps, which yields the claimed running time.

If $G$ is basic, then there is nothing to prove.
If $G$ is not proper, a reduction exists by Lemma~\ref{lemma-no5}.
If $G$ is proper and fails to satisfy (CT1), the existence of a reduction follows from Lemma~\ref{lemma-no2in7} (note that
$G$ cannot have a cycle of length at most six containing a vertex of degree two since it is proper).
If $G$ is proper and does not satisfy (CT2), then a reduction exists by Lemma~\ref{lemma-no26}, and
if it does not satisfy (CT3), then a reduction exists by Lemma~\ref{lemma-cutsin7}.
Finally, if $G$ is proper and fails to not satisfy (CT4),
then a reduction exists by one of Lemmas~\ref{lemma-6-mainred}, \ref{lemma-6no2e}, \ref{lemma-6adjcuts} or \ref{lemma-6oppcuts}.
\end{proof}

\section{Main result}\label{sec-proof}

We need few additional results before we can prove Theorem~\ref{thm-main}.
The first concerns the structure of cycles passing through vertices of a cycle of length six
in a clean $2$-connected subcubic graph.
Let $v$ be a vertex of degree three in a graph $G$, and let $x_1$, $x_2$ and $x_3$ be its neighbors.
The \emph{type} of $v$ is the triple $(\ell_1, \ell_2, \ell_3)$ such that
$\ell_1$, $\ell_2$ and $\ell_3$ are the lengths of shortest cycles containing paths $x_1vx_2$, $x_1vx_3$ and $x_2vx_3$.
In our consideration, the order of the coordinates of the triple will be irrelevant,
so we will always assume that the lengths satisfy that $\ell_1\le\ell_2\le\ell_3$.
A type $(\ell'_1,\ell'_2,\ell'_3)$ \emph{dominates} the type $(\ell_1, \ell_2, \ell_3)$
if $\ell'_i\ge \ell_i$ for every $i=1,2,3$.
If $K$ is a cycle in a graph $G$ and each vertex of $K$ has degree three,
then the \emph{type} of the cycle $K$ is the multiset of the types of the vertices of $K$.
Finally, a multiset $M_1$ of types \emph{dominates} a multiset $M_2$ types
if there exists a bijection between the types contained in $M_1$ and $M_2$ such that
each type of $M_1$ is dominates the corresponding type in $M_2$.

We can now prove the following lemma (note that all vertices of the cycle $K$ in the lemma
must have degrees three since $G$ is assumed to be clean).

\begin{lemma}\label{lemma-6types}
Let $G$ be a clean $2$-connected subcubic graph and let $K=v_1v_2\ldots v_6$ be a cycle of length six in $G$.
If $K$ is not a $\theta$-cycle, then the type of $K$ dominates at least one of the following multisets:
\begin{itemize}
\item $\{(6,7,7),\;(6,7,7),\;(6,8,8),\;(6,8,8),\;(6,8,8),\;(6,8,8)\}$,
\item $\{(6,7,7),\;(6,7,8),\;(6,7,8),\;(6,8,8),\;(6,8,8),\;(6,8,8)\}$, or
\item $\{(6,7,7),\;(6,7,8),\;(6,7,9),\;(6,7,9),\;(6,8,8),\;(6,8,8)\}$.
\end{itemize}
\end{lemma}

\begin{proof}
Let $x_i$ be the neighbor of $v_i$ outside of $K$, $i=1,\ldots,6$.
Since $G$ is clean, the cycle $K$ satisfies one of the four conditions in (CT4).
As $K$ is not a $\theta$-cycle, it must satisfy (CT4)(b), (CT4)(c) or (CT4)(d).
We analyze each of these three cases separately. 

Suppose that the cycle $K$ satisfies (CT4)(b),
i.e., each edge $v_ix_i$, $i=1,\ldots,6$,  is contained in a $2$-edge-cut but no two of them together form a $2$-edge-cut.
Let $K'$ be a cycle in $G$ containing the edge $v_1x_1$.
If the intersection of $K$ and $K'$ is not a path, then the length of $K'$ is at least ten by (CT2).
In the rest, we assume that the intersection of $K$ and $K'$ is a path and that
the cycles $K$ and $K'$ share a path $v_1v_2\ldots v_k$.
If $k=2$, then the length of $K'$ is at least $8$ by (CT3) since both $v_1x_1$ and $v_2x_2$ are contained in a $2$-edge-cut.
If $x_1$ or $x_k$ has degree two, then the length of $K'$ is also at least $8$ by (CT1).
Hence, we assume that $k\ge 3$ and that both the vertices $x_1$ and $x_k$ have degree three.

Let $C_1$ and $C_2$ be the blocks of $G-V(K)$ containing the vertices $x_1$ and $x_k$, respectively, and
let $e_1$ and $e_2$ be the cut-edges of $G-V(K)$ that are contained in $K'$ and
that are incident with $C_1$ and $C_2$, respectively.
Note that $C_1$ and $C_2$ are vertex-disjoint and $e_1\not=e_2$
since the edges $v_1x_1$ and $v_kx_k$ do not form a $2$-edge-cut by (CT4)(b).
In addition, $e_1$ is not incident with $x_1$ and $e_2$ is not incident with $x_k$
since the vertices $x_1$ and $x_k$ have degree three and $G$ is $2$-connected.
We conclude that the cycle $K'$ has at least five vertices outside the cycle $K$:
the vertices $x_1$, $x_k$ and the end vertices of $e_1$ and $e_2$.
Hence, the length of the cycle $K'$ is at least $k+5\ge 8$.

Since $K'$ was an arbitrary cycle containing the edge $v_1x_1$ and
the symmetric argument applies to each of the edges $v_ix_i$, $i=1,\ldots,6$,
we conclude that the type of each vertex of $K$ dominates $(6,8,8)$.
In particular, the type of $K$ dominates the first multiset from the statement of the lemma.

Suppose next that $K$ satisfies (CT4)(c),
i.e., each edge $v_ix_i$, $i=1,\ldots,6$, is contained in a $2$-edge-cut, and
there exists exactly one pair $i$ and $j$ with $1\le i<j\le 6$ such that
the edges $v_ix_i$ and $v_jx_j$ form a $2$-edge-cut, and this pair satisfies $j-i=3$.
By symmetry, we can assume that the edges $v_1x_1$ and $v_4x_4$ form a $2$-edge-cut.
Let $K'$ be an arbitrary cycle containing an edge $v_ix_i$ for $i=1,\ldots,6$.
The length of $K'$ is at least seven by (CT2),
which implies that the type of $v_i$ dominates $(6,7,7)$.
If $i\in\{2,3,5,6\}$, then the arguments presented in the analysis of the case (CT4)(b)
yield that the length of $K'$ is at least eight, i.e., the type of $v_i$ dominates $(6,8,8)$.
We conclude that the type of $K$ dominates the first multiset from the statement of the lemma.

Finally, suppose that $K$ satisfies (CT4)(d), i.e.,
the edge $v_1x_1$ is not contained in a $2$-edge-cut while
each of the edges $v_ix_i$, $i=2,\ldots,6$, is a contained in a $2$-edge-cut,
there exists a partition $A$ and $B$ of the vertices of $G-V(K)$
such that $x_1,x_2,x_6\in A$, $x_3,x_4,x_5\in B$,
there is exactly one edge between $A$ and $B$, and
both $A$ and $B$ induce connected subgraphs of $G-V(K)$.
Note that the structure of $G-V(K)$ implies that
no two of the edges $v_1x_1,\ldots,v_6x_6$ form a $2$-edge-cut.

Let $K'$ be an arbitrary cycle containing an edge incident with the cycle $K$.
The length of $K'$ is at least seven by (CT2).
If the intersection of the cycles $K$ and $K'$ is not a path, then the length of $K'$ is at least eight.
If the cycle $K'$ does not contain the edge $v_1x_1$,
then the analysis of the case (CT4)(b) yields that the length of $K'$ is at least eight.
Hence, the length of $K'$ is at least eight unless
$K'$ contains the edge $v_1x_1$ and the intersection of $K$ and $K'$
is a path $v_1v_2\ldots v_k$ with $k\le 4$ (if $k=5$ or $k=6$, then the length of $K'$ is at least nine by (CT2)).

If every cycle of length seven intersects the cycle $K$ only in two vertices,
then the type of $v_1$ dominates $(6,7,7)$,
the types of $v_2$ and $v_6$ dominate $(6,7,8)$, and
the types of $v_3$, $v_4$, and $v_5$ dominate $(6,8,8)$.
Consequently, the type of $K$ dominates the second multiset from the statement of the lemma.

In the rest of the proof, we assume that there exists a cycle $K'$ of length seven such that
the intersection of $K$ and $K'$ is a path $v_1v_2\ldots v_k$ with $k=3$ or $k=4$.
Let $C$ be the block of $G-V(K)$ containing the vertex $x_k$, and
let $e=zz'$ be the cut-edge incident with $C$ that is contained in $K'$.
Note that $V(C)\subset B$, in particular, $x_1\not\in V(C)$.
Since each of the edges $v_3x_3$, $v_4x_4$ $v_5x_5$ is contained in a $2$-edge-cut and
the subgraph of $G-V(K)$ induced by $B$ is connected, both the end-vertices of $e$ are in $B$,
i.e., $e$ is not the edge between $A$ and $B$.
By (CT1), the degrees of $x_k$ is three, which implies that $e$ is not incident with $x_k$.
We conclude that $k=3$: otherwise, $K'$ contains the four vertices $v_1,\ldots,v_4$,
the vertices $x_1$ and $x_4$, and the two end-vertices of $e$.
Moreover, the cycle $K'$ is the cycle $v_1v_2v_3x_3zz'x_1$.
Note that since $k=3$, the type of $v_4$ dominates $(6,8,8)$.

Since both the end-vertices $z$ and $z'$ of the edge $e$ are contained in $B$,
the edge $z'x_1$ is the unique edge between $A$ and $B$.
Since the edge $v_3x_3$ is contained in a $2$-edge-cut, $G$ is $2$-connected and the degree of $x_3$ is three,
if follows that the edges $v_3x_3$ and $zz'$ form a $2$-edge-cut in $G$.
If $G$ had a cycle of length seven passing through the vertex $v_5$,
then the symmetric argument would yield that the edges $v_5x_5$ and $z''z'$ form a $2$-edge-cut in $G$,
where $z''$ is the neighbor of $z'$ different from $z$ and $x_1$.
Since this is impossible since the edge $v_4x_4$ is contained in a $2$-edge-cut and
the subgraph of $G-V(K)$ induced by $B$ is connected,
we conclude that the vertex $v_5$ is contained in no cycle of length seven.
Hence, we have established that the type of $v_5$ dominates $(6,8,8)$.
Also note that the type of $v_3$ dominates $(6,7,8)$
since any cycle of length seven containing $v_3$ contains the path $v_2v_3x_3$.

Consider now a cycle $K''$ in $G$ containing the path $v_3v_2x_2$.
If $K''$ contains only the vertices from $V(K)\cup A$,
then $K''$ contains at least five vertices of $K$ and at least four vertices of $A$;
otherwise, there would be a cycle of length six intersecting the cycle $K$, which is excluded by (CT2).
If the cycle $K''$ contains some vertices from the set $B$,
then it contains at least two vertices of the cycle $K$,
five vertices of $A$ (otherwise, the path of $K''$ from $x_2$ to $x_1$ together with the path $x_2v_2v_1x_1$
would form a cycle of length six intersecting $K$) and
two vertices in $B$ (the vertices $x_3$ and $z'$ cannot coincide since the edge $v_3x_3$ is contained in a $2$-edge-cut).
In both cases the length of $K''$ is at least nine.
We conclude that the type of the vertex $v_2$ dominates $(6,7,9)$.
The symmetric argument yields that the type of $v_6$ dominates $(6,7,9)$.
Since the type of $v_1$ dominates $(6,7,7)$, it follows that the type of $K$
dominates the third multiset from the statement of the lemma.
\end{proof}

The following lemma follows from the description of the perfect matching polytope by Edmonds~\cite{edmonds1965maximum} and
the fact that the perfect matching polytope has a strong separation oracle~\cite{padberg1982odd};
see e.g.~\cite{grotschel2012geometric} for further details.

\begin{lemma}\label{lemma-2m}
There exists a polynomial-time algorithm that
for a given cubic $2$-connected $n$-vertex graph
outputs a collection of $m\le n/2+2$ perfect matchings $M_1,\ldots,M_m$ and non-negative coefficients $a_1,\ldots,a_m$ such that
$a_1+\cdots+a_m=1$ and
$$\sum_{i=1}^m a_i\chi_{M_i}=(1/3,\ldots,1/3)\in\RR^{E(G)}\;,$$
where $\chi_{M_i}\in\RR^{E(G)}$ is the characteristic vector of $M_i$.
\end{lemma}

Lemma~\ref{lemma-2m} gives the following.

\begin{lemma}\label{lemma-2fs}
There exists a polynomial-time algorithm that
for a given $2$-connected $n$-vertex subcubic graph
outputs a collection of $m\le n/2+2$ spanning Eulerian subgraphs $F_1,\ldots,F_m$ and
probabilities $p_1,\ldots,p_m\ge 0$, $p_1+\cdots+p_m=1$ that satisfy the following.
If a spanning Eulerian subgraph $F$ is equal to $F_i$ with probability $p_i$, $i=1,\ldots,m$,
then $\Prob[e\in E(F)]=2/3$.
In particular, a vertex of degree three is contained in a cycle of $F$ with probability one and
a vertex of degree two is isolated with probability $1/3$.
\end{lemma}

\begin{proof}
Let $G$ be the input $2$-connected $n$-vertex subcubic graph, and
let $G'$ be the $2$-connected cubic graph obtained from $G$ by suppressing all vertices of degree two.
Apply the algorithm from Lemma~\ref{lemma-2m} to $G'$ to get a collection of $m$ perfect matchings $M_1,\ldots,M_m$ and
non-negative coefficients $a_1,\ldots,a_m$ with the properties stated in the lemma. Note that $m\le n/2+2$.
Let $F'_i$ be the $2$-factor of $G'$ consisting of the edges not contained in $M_i$, and
let $F_i$ be the spanning Eulerian subgraph of $G$ consisting of the edges contained in paths corresponding to the edges of $F'_i$,
$i=1,\ldots,m$.
It is easy to see that the lemma holds for $F_1,\ldots,F_m$ with $p_i=a_i$, $i=1,\ldots,m$.
\end{proof}

We now combine Lemmas~\ref{lemma-6types} and~\ref{lemma-2fs}.

\begin{lemma}\label{lemma-expected}
There exists a polynomial-time algorithm that given a clean $2$-connected subcubic graph $G$
outputs a spanning Eulerian subgraph $F$ of $G$ such that $$\exc(F)\le \frac{2n(G)+2n_2(G)}{7}\;.$$
\end{lemma}

\begin{proof}
We first apply the algorithm from Lemma~\ref{lemma-2fs} to get
a collection of $m\le n/2+2$ spanning Eulerian subgraphs $F_1,\ldots,F_m$ and probabilities $p_1,\ldots,p_m$.
We show that
\begin{equation}
\EE\exc(F)\le \frac{2n(G)+2n_2(G)}{7}\;,\label{eq-2fs}
\end{equation}
which implies the statement of the lemma since the number of the subgraphs $F_1,\ldots,F_m$ is linear in $n$ and
the excess of each them can be computed in linear time.
In particular, the algorithm can output the subgraph $F_i$ with the smallest $\exc(F_i)$.

We now show that (\ref{eq-2fs}) holds.
We apply a double counting argument, which we phrase as a discharging argument.
At the beginning, we assign each vertex of degree three charge of $2/7$ and
to each vertex of degree two charge of $4/7$. Let $c_1(v)$ be the initial charge of a vertex $v$.
Note that the sum of the initial charges of the vertices
is the right side of the inequality (\ref{eq-2fs}).

We next choose a random spanning Eulerian subgraphs $F$ among the subgraphs $F_1,\ldots,F_m$
with probabilities given by $p_1,\ldots,p_m$.
The charge of each vertex that is isolated in $F$ is decreased by one unit, and
the charge of each vertex contained in a cycle of length $k$ by $2/k$ units.
Let $c_2(v)$ be the new charge of a vertex $v$.
Observe that the total decrease of charge of the vertices is equal to $\exc(F)$,
i.e.,
$$\exc(F)=\sum_{v\in V(G)} c_1(v)-c_2(v)\;.$$
Hence, it is enough to prove that
\begin{equation}
\EE\,\sum_{v\in V(G)} c_2(v)\ge 0.
\label{eq:Ec1}
\end{equation}
To prove (\ref{eq:Ec1}), we consider the expectation of $c_2(v)$ for individual vertices $v$ of $G$.

If $v$ is a vertex of $G$ of degree two,
then every cycle of $G$ that contains $v$ has length at least eight by (CT1).
With probability $1/3$, the vertex $v$ is isolated and looses one unit charge;
with probability $2/3$, it is contained in a cycle and looses at most $2/8=1/4$ units of charge.
We conclude that
$$\EE c_2(v)\ge \frac{4}{7}-\frac{1}{3}-\frac{2}{3}\cdot\frac{1}{4}=\frac{1}{14}>0\;.$$
If $v$ is a vertex of $G$ of degree three with type $(\ell_1,\ell_2,\ell_3)$,
we proceed as follows.
Since each edge incident with $v$ is contained in $F$ with probability $2/3$,
$v$ is contained in a cycle of $F$ with a particular pair of its neighbors with probability $1/3$.
It follows that the expected value of $c_1(v)$ is at least
$$\EE c_2(v)\ge\frac{2}{7}-\frac{1}{3}\left(\frac{2}{\ell_1}+\frac{2}{\ell_2}+\frac{2}{\ell_3}\right)\;.$$
Since $G$ is clean, the type of $v$ dominates $(6,6,6)$.
If the type of $v$ dominates $(7,7,7)$, then $\EE c_1(v)\ge 0$.
Hence, we focus on vertices contained in cycles of length six in $G$ in the rest of the proof.

Let $K=v_1\ldots v_6$ be a cycle of length six in $G$.
Since $G$ is clean, each vertex of $K$ has degree three.
Suppose that $K$ is not a $\theta$-cycle.
By (CT2), $K$ is disjoint from all other cycles of length six in $G$.
Observe that
\begin{itemize}
\item if the type of $v_i$ dominates $(6,8,8)$, then $\EE c_2(v_i)\ge \frac{1}{126}$,
\item if the type of $v_i$ dominates $(6,7,7)$, then $\EE c_2(v_i)\ge -\frac{1}{63}$,
\item if the type of $v_i$ dominates $(6,7,8)$, then $\EE c_2(v_i)\ge -\frac{1}{252}$, and
\item if the type of $v_i$ dominates $(6,7,9)$, then $\EE c_2(v_i)\ge \frac{1}{189}$.
\end{itemize}
Since the type of the cycle $K$ dominates one of the three multisets listed in Lemma~\ref{lemma-6types},
it holds that
$$\EE c_2(v_1)+\cdots+c_2(v_6)\ge 0\;.$$

It remains to analyze the case that $K$ is a $\theta$-cycle.
By symmetry, we can assume that the vertices $v_1$ and $v_4$ are its poles.
Let $x_i$ be the neighbor of $v_i$ outside of $K$, $i=1,\ldots,6$.
Further, let $P=x_6v_6v_1v_2x_2$, $P_1=x_6v_6v_1$ and $P_2=x_2v_2v_1$.
Since each of the paths $P_1$ and $P_2$ is contained in $F$ with probability $1/3$,
the subgraph $F$ contains the path $P$ with probability at most $1/3$;
let $p$ be this probability.
Since $G$ is clean (and so proper), the distance between $x_2$ and $x_3$ in $G-V(K)$ is at least three;
likewise, the distance between $x_5$ and $x_6$ in $G-V(K)$ is at least three.
Hence, any cycle containing $P_1$ or $P_2$ has length at least $10$, and
any cycle containing $P$ has length at least $14$.
Since $F$ contains the path $P$ with probability $p$,
the path $P_1$ but not $P$ with probability $1/3-p$,
the path $P_2$ but not $P$ with probability $1/3-p$, and
neither $P_1$ nor $P_2$ with probability $1/3+p$,
it follows that
$$\EE c_2(v_1)=\frac{2}{7}-p\cdot\frac{1}{7}-2\left(\frac{1}{3}-p\right)\cdot\frac{1}{5}-\left(\frac{1}{3}+p\right)\cdot\frac{1}{3}=
  \frac{13}{315}-\frac{8}{105}p\ge\frac{1}{63}\;.$$ 
The symmetric argument yields that $\EE c_2(v_4)\ge\frac{1}{63}$.
Since every cycle in $G$ containing the path $P_2$ has length at least $10$,
the type of $v_2$ dominates $(6,6,10)$ and thus $\EE c_2(v_2)\ge -\frac{1}{315}$.
The same holds for vertices $v_3$, $v_5$ and $v_6$.

Let $Q_1$ be the set of all poles of $\theta$-cycles in $G$, and
let $Q_2$ be the set of vertices contained in $\theta$-cycle that are not a pole of a (possibly different) $\theta$-cycle.
Since each vertex of $Q_2$ has a neighbor in $Q_1$, it follows $|Q_2|\le 3|Q_1|$.
The previous analysis yields that
$$\EE\sum_{v\in Q_1\cup Q_2} c_2(v)\ge |Q_1|\left(\frac{1}{63}-3\cdot\frac{1}{315}\right)=\frac{2}{315}|Q_1| \ge 0.$$
Since the set $Q_1\cup Q_2$ and the vertex set of cycles of length six that are not $\theta$-cycles are disjoint,
the inequality (\ref{eq:Ec1}) follows.
\end{proof}

We are ready to prove Theorems~\ref{thm-main} and~\ref{thm-alg}.

\begin{proof}[Proof of Theorem~\ref{thm-main}]
By Observation~\ref{obs-exc},
it is enough to construct a spanning Eulerian subgraph $F$ of $G$ with
$$\exc(F)\le \frac{2(n(G)+n_2(G))}{7}+1.$$
If $G$ is basic, such a subgraph $F$ exists by Observation~\ref{obs-nontarg}, and
can easily be constructed in polynomial time.
If $G$ is not basic,
we can find a reduction $G'$ of $G$ that is either basic or clean in polynomial time by Theorem~\ref{cor-summary}.

If $G'$ is basic, then we find a spanning Eulerian subgraph with
$$\exc(F')\le \frac{2(n(G')+n_2(G'))}{7}+1$$
as in the case when $G$ itself is basic.
If $G'$ is clean, then Lemma~\ref{lemma-expected} yields that
we can construct in polynomial time a spanning Eulerian subgraph $F'$ of $G'$ such that
$$\exc(F')\le \frac{2(n(G')+n_2(G'))}{7}.$$
Since $G'$ is a reduction of $G$, we can find in polynomial time a spanning Eulerian subgraph $F$ of $G$ such that
$$\exc(F)\le \exc(F')+\frac{\delta(G,G')}{4}\le \exc(F')+\frac{2\delta(G,G')}{7}\le \frac{2(n(G')+n_2(G'))}{7}+1\;,$$
which finishes the proof of the theorem.
\end{proof}

\begin{proof}[Proof of Theorem~\ref{thm-alg}]
Let $G$ be an input cubic graph and let $n$ be the number of its vertices.
We assume that $G$ is connected since $G$ would not have a TSP walk otherwise.
Let $F$ be the set of bridges of $G$, which can be found in linear time using the standard algorithm based on DFS.
Further, let $G'$ be the graph obtained from $G$ by removing the edges of $F$, and
let $n_0$ and $n_2$ be the number of its vertices of degree zero and two, respectively.
Note that $G'$ has no vertices of degree one
since if two edges incident with a vertex $v$ in a cubic graph are bridges,
then the third edge incident with $v$ is also a bridge.
Finally, let $k$ be the number of non-trivial components of $G'$,
i.e., the components of $G'$ that are not formed by a single vertex.
Observe that the number of vertices of degree two in $G'$ is at most $2k-2$,
i.e., $n_2\le 2k-2$.

We next apply the algorithm from Theorem~\ref{thm-main} to each non-trivial component of $G'$, and
obtain a collection of $k$ TSP walks such that the sum of their lengths is at most
$$\frac{9}{7}(n-n_0)+\frac{2}{7}n_2-k\;\mbox{.}$$
These $k$ TSP walks can be connected by traversing each of the edges of $F$ twice,
which yields a TSP walk in $G$ of total length at most
\begin{equation}
\frac{9}{7}(n-n_0)+\frac{2}{7}n_2-k+2|F|\;\le\;\frac{9}{7}(n-n_0)+2|F|\;\mbox{.}\label{eq-alg-1}
\end{equation}
The inequality in (\ref{eq-alg-1}) follows from the inequality $n_2\le 2k-2$,
which we have observed earlier in the proof.
Since any TSP walk in $G$ must have length at least $(n-n_0)+2|F|$,
the upper bound in (\ref{eq-alg-1}) on the length of the constructed TSP walk
is at most the multiple of $9/7$ of the length of the optimal TSP walk in $G$,
which yields the desired approximation factor of the algorithm.
\end{proof}

\section{Lower bounds}\label{sec-lb}

In this section, we provide two constructions of 2-connected subcubic graphs that illustrate that
the bound claimed in Conjecture~\ref{conj-main} would be the best possible.
The constructions are based on two operations that we analyze in Lemmas~\ref{lemma-drepl} and~\ref{lemma-qrepl}.

\begin{figure}
\begin{center}
\includegraphics{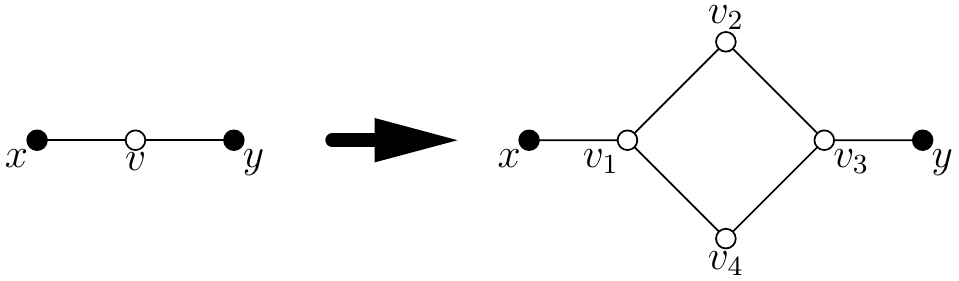}
\end{center}
\caption{Replacing a vertex of degree two with a cycle of length four in Lemma~\ref{lemma-drepl}.}\label{fig-DGv}
\end{figure}

\begin{lemma}\label{lemma-drepl}
Let $G$ be a 2-connected subcubic graph,
let $v$ be a vertex of $G$ that has exactly two neighbors, and
let $x$ and $y$ be its two neighbors.
Further, let $G'$ be the graph obtained from $G$ by removing the vertex $v$,
adding a cycle $v_1v_2v_3v_4$ and edges $xv_1$ and $yv_3$ as in Figure~\ref{fig-DGv}.
The graph $G'$ is a 2-connected subcubic graph and
it holds that $n(G')=n(G)+3$, $n_2(G')=n_2(G)+1$ and $\minexc(G')=\minexc(G)+1$.
\end{lemma}

\begin{proof}
It is clear that $G'$ is a 2-connected subcubic graph such that
$n(G')=n(G)+3$ and $n_2(G')=n_2(G)+1$.
So, we need to show that $\minexc(G')=\minexc(G)+1$.
We start with showing that $\minexc(G')\le\minexc(G)+1$.
Let $F$ be a spanning Eulerian subgraph of $G$ with $\exc(F)=\minexc(G)$.
We now construct a spanning Eulerian subgraph $F'$ of $G'$.
If $v$ is an isolated vertex in $F$, then $F'$ contains all the edges of $F$ and the cycle $v_1v_2v_3v_4$.
Note that $c(F')=c(F)+1$ and $i(F')=i(F)-1$,
Otherwise, $v$ has degree two in $F$ and we let $F'$ to contain the edges $xv_1$, $v_1v_2$, $v_2v_3$, $v_3y$ and
all the edges of $F$ except for $vx$ and $vy$.
In this case, we have that $c(F')=c(F)$ and $i(F')=i(F)+1$.
In both cases, we get that $\exc(F')=\exc(F)+1=\minexc(G)+1$,
which implies that $\minexc(G')\le\minexc(G)+1$.

We next prove that $\minexc(G)\le\minexc(G')-1$.
Consider a spanning Eulerian subgraph $F'$ of $G'$ with $\exc(F')=\minexc(G')$.
We reverse the transformation described in the previous paragraph.
By symmetry, we can assume that $F'$ contains either the cycle $v_1v_2v_3v_4$ or the path $xv_1v_2v_3y$.
In the former case, let $F$ be the spanning Eulerian subgraph of $G$ containing all the edges of $F'$
except for the edges of the cycle $v_1v_2v_3v_4$.
In the latter case, let $F$ be the spanning Eulerian subgraph of $G$ containing the edges $vx$, $vy$ and
all the edges of $F'$ except for the edges $xv_1$, $v_1v_2$, $v_2v_3$, $v_3y$.
In both cases, it holds that $\exc(F)=\exc(F')-1$,
which implies that $\minexc(G)\le\minexc(G')-1$ as desired.
\end{proof}

Repeated applications of the operation described in Lemma~\ref{lemma-drepl}
starting with the graph $K_{2,3}$ yields the following.

\begin{proposition} \label{prop-drepl}
For every integer $n\ge 5$, $n\equiv 2\;{\rm mod}\; 3$,
there exists a $2$-connected subcubic $n$-vertex graph $G$ such that
$$\minexc(G)=\frac{n(G)+n_2(G)}{4}+1.$$
\end{proposition}

The second operation is more involved.
A \emph{diamond} in a graph $G$ is an induced subgraph isomorphic to $K_4^-$,
i.e., the graph $K_4$ with one edge removed.

\begin{figure}
\begin{center}
\includegraphics[width=120mm]{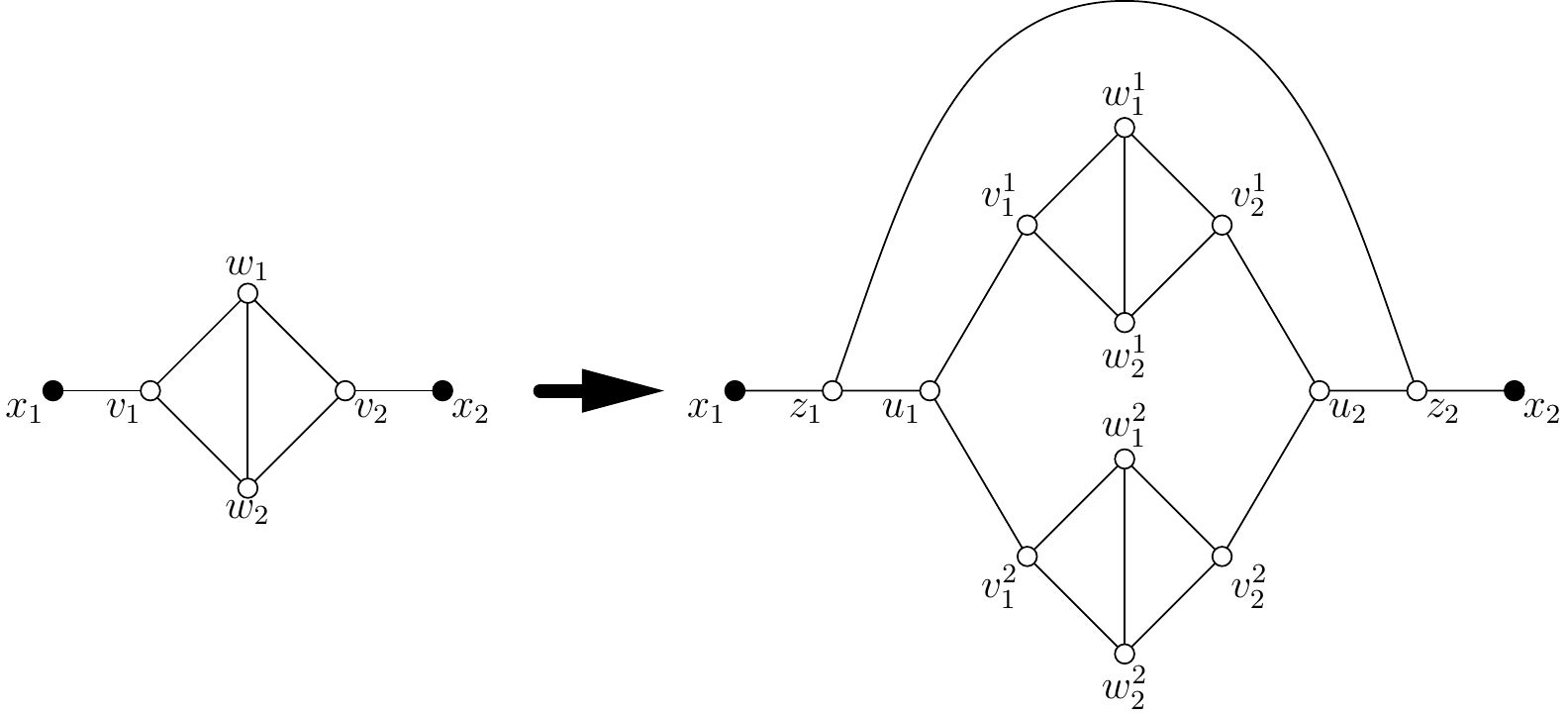}
\end{center}
\caption{The operation of replacing a diamond analyzed in Lemma~\ref{lemma-qrepl}.}\label{fig-QGK}
\end{figure}

\begin{lemma}\label{lemma-qrepl}
Let $G$ be a 2-connected cubic graph containing a diamond $D$.
Let $v_1$, $v_2$, $w_1$ and $w_2$ be the vertices of the diamond as depicted in Figure~\ref{fig-QGK}, and
let $x_1$ and $x_2$ be the neighbors of $v_1$ and $v_2$ outside of the diamond $D$.
Further, let $G'$ be the graph obtained from $G$ by removing the vertices of the diamond $D$ and
inserting the subgraph depicted in Figure~\ref{fig-QGK}.
The graph $G'$ is a 2-connected cubic graph with $n(G')=n(G)+8$ and $\minexc(G')=\minexc(G)+2$.
Moreover, the graph $G'$ contains at least two diamonds.
\end{lemma}

\begin{proof}
As in the proof of Lemma~\ref{lemma-drepl},
the only non-trivial assertion of the lemma is that $\minexc(G')=\minexc(G)+2$.
Let the labels of the vertices be as in Figure~\ref{fig-QGK}. 
We start with showing $\minexc(G')\le\minexc(G)+2$.
Consider a spanning Eulerian subgraph $F$ of $G$ with $\exc(F)=\minexc(G)$.
By symmetry, we can assume that the subgraph $G$ contains either the path $x_1v_1w_1w_2v_2x_2$ or the cycle $v_1w_1v_2w_2$.
In the former case, let $F'$ be the spanning subgraph of $G'$ that
contains the path $x_1z_1u_1v^1_1w^1_1w^1_2v^1_2u_2z_2x_2$ and the cycle $v^2_1w^2_1v^2_2w^2_2$
instead of the path $x_1v_1w_1w_2v_2x_2$.
In the latter case, $F'$ is the spanning subgraph of $G'$ that
contains the cycles $z_1u_1v^1_1w^1_1w^1_2v^1_2u_2z_2$ and $v^2_1w^2_1v^2_2w^2_2$
instead of the cycle $v_1w_1v_2w_2$.
In both cases, it holds that $c(F')=c(F)+1$ and $i(F')=i(F)$,
which implies that $\minexc(G')\le\exc(F')=\exc(F)+2=\minexc(G)+2$.

We next prove the opposite inequality $\minexc(G)\le\minexc(G')-2$.
Let $F'$ be a spanning Eulerian subgraph of $G'$ with $\exc(F')=\minexc(G')$.
A simple case analysis using that $F'$ has the minimum possible excess yields that
we can assume that $F'$ contains either the path $x_1z_1u_1v^1_1w^1_1w^1_2v^1_2u_2z_2x_2$ and the cycle $v^2_1w^2_1v^2_2w^2_2$ or
the cycles $z_1u_1v^1_1w^1_1w^1_2v^1_2u_2z_2$ and $v^2_1w^2_1v^2_2w^2_2$.
In both cases, we can reverse the operation described in the previous paragraph
to get a spanning Eulerian subgraph $F$ of $G$ with $\exc(F)=\exc(F')-2$.
It follows that $\minexc(G)\le\exc(F)=\exc(F')-2=\minexc(G')-2$ as desired.
\end{proof}

Consider the cubic graph formed by two diamonds and two edges joining the vertices of degree two in different diamonds, and
repeatedly apply the operation described in Lemma~\ref{lemma-qrepl}.

\begin{proposition} \label{prop-qrepl}
For every integer $n\ge 8$, $n\equiv 0\;{\rm mod}\; 8$,
there exists a $2$-connected cubic $n$-vertex graph $G$
with $\minexc(G)=n/4$.
\end{proposition}

Propositions~\ref{prop-drepl} and~\ref{prop-qrepl}, and Observation~\ref{obs-exc} yield that
neither the coefficient $5/4$ nor the coefficient $1/4$ in Conjecture~\ref{conj-main} can be improved.
Indeed, for every $\alpha<5/4$,
there exist infinitely many 2-connected cubic graphs $G$ with $\tsp(G)>\alpha n(G)+o(n(G))$ by Proposition~\ref{prop-qrepl}.
Likewise, for every $\beta<1/4$,
there exist infinitely many 2-connected subcubic graphs $G$ with $\tsp(G)>\frac{5}{4}n(G)+\beta n_2(G)+o(n(G))$.
While neither of the two coefficient in Conjecture~\ref{conj-main} can be improved in general,
it may be possible to prove better bounds under some additional structural assumptions.
In particular, Conjecture~\ref{conj-main} asserts that $\tsp(G)\le 3n(G)/2$ for 2-connected subcubic graph
while Boyd et al.~\cite{boyd} proved that $\tsp(G)\le 4n(G)/3$ for such graphs $G$,
which is tight up to an additive constant.

\bibliographystyle{siam}
\bibliography{cubictsp}

\end{document}